\documentclass[journal,12pt,onecolumn,draftcls]{IEEEtran}
\usepackage{cite}
\usepackage{algorithm}
\usepackage{algpseudocode}
\usepackage{algorithmicx}
\usepackage{graphicx}
\usepackage{textcomp}
\usepackage[latin1]{inputenc}

\usepackage{amssymb,amsmath,amsfonts,cases,mathtools}
\usepackage{cases}
\usepackage{microtype}
\usepackage{url}
\usepackage{hyperref}
\usepackage{xifthen}
\usepackage{xpatch}
\usepackage{amsthm}
\usepackage[dvipsnames]{xcolor}
\usepackage{tikz}
\usepackage{subfig}
\usepackage{lipsum,booktabs}
\usepackage{pifont}
\usepackage{xcolor}
\usepackage{adjustbox}

\newcommand{\cmark}{\textcolor{green!80!black}{\ding{51}}}
\newcommand{\xmark}{\textcolor{red}{\ding{55}}}

\newcommand{\cE}{\mathcal{E}}
\newcommand{\cG}{\mathcal{G}}
\newcommand{\cW}{\mathcal{W}}

\newcommand{\cN}{\mathcal{N}}
\newcommand{\cF}{\mathcal{F}}
\newcommand{\cU}{\mathcal{U}}

\newcommand{\mc}{\mathcal}

\newcommand{\R}{\mathbb{R}}
\newcommand{\N}{\mathbb{N}}

\newcommand{\x}{x}

\newcommand{\w}{w}
\newcommand{\z}{z}
\newcommand{\y}{y}
\newcommand{\X}{X}

\newcommand{\iter}{t}
\newcommand{\iterp}{t+1}
\newcommand{\xtp}{\x^{\iterp}}

\newcommand{\wtp}{\w^{\iterp}}
\newcommand{\ztp}{\z^{\iterp}}
\newcommand{\ytp}{\y^{\iterp}}

\newcommand{\xt}{\x^{\iter}}

\newcommand{\wt}{\w^{\iter}}
\newcommand{\zt}{\z^{\iter}}
\newcommand{\yt}{\y^{\iter}}

\newcommand{\Wd}{\cW_d}
\newcommand{\Wm}{\cW_m}
\newcommand{\xstar}{\x^\star}
\newcommand{\xistar}{\x^\star_i}
\newcommand{\xstari}{\x^\star_{-i}}

\newcommand{\xitp}{\x_{i}^{\iterp}}

\newcommand{\zitp}{\z_{i}^{\iterp}}
\newcommand{\yitp}{\y_{i}^{\iterp}}
\newcommand{\xit}{\x_{i}^{\iter}}

\newcommand{\zit}{\z_{i}^{\iter}}
\newcommand{\yit}{\y_{i}^{\iter}}
\newcommand{\xjt}{\x_{j}^{\iter}}

\newcommand{\zjt}{\z_{j}^{\iter}}
\newcommand{\yjt}{\y_{j}^{\iter}}

\newcommand{\tz}{\tilde{\z}}

\newcommand{\tw}{\tilde{\w}}

\newcommand{\tztp}{\tz^{\iterp}}

\newcommand{\twtp}{\tw^{\iterp}}
\newcommand{\psitp}{\psi^{\iterp}}

\newcommand{\tzt}{\tz^{\iter}}

\newcommand{\twt}{\tw^{\iter}}

\newcommand{\psit}{\psi^{\iter}}

\newcommand{\bz}{\bar{\z}}
\newcommand{\by}{\bar{\y}}

\newcommand{\pz}{\z_\perp}
\newcommand{\py}{\y_\perp}

\newcommand{\bzt}{\bz^{\iter}}
\newcommand{\bztp}{\bz^{\iterp}}
\newcommand{\byt}{\by^{\iter}}

\newcommand{\pzt}{\pz^{\iter}}
\newcommand{\pztp}{\pz^{\iterp}}
\newcommand{\pyt}{\py^{\iter}}

\newcommand{\tpyt}{\tilde{\y}_\perp^\iter}

\newcommand{\xmi}{\x_{-i}}
\newcommand{\xmit}{\xmi^\iter}
\newcommand{\sigmai}{\sigma_{-i}}
\newcommand{\xibr}{x_{i,\text{br}}}

\newcommand{\chit}{\chi^{\iter}}
\newcommand{\chitp}{\chi^{\iterp}}
\newcommand{\chistar}{\chi^{\star}}

\DeclareMathOperator*{\argmin}{arg\,min} %
\newcommand{\Ji}{J_i}

\newcommand{\F}{F}

\newcommand{\tFi}{\tilde{F}_i}
\newcommand{\Gxi}{G_{\x,i}}
\newcommand{\Gli}{G_{\lambda,i}}
\newcommand{\tF}{\tilde{F}}
\newcommand{\Gx}{G_{\x}}
\newcommand{\Gl}{G_{\lambda}}

\newcommand{\phii}{\phi_i}
\newcommand{\phij}{\phi_j}

\newcommand{\oned}{\mathbf{1}_{N,d}}
\newcommand{\onem}{\mathbf{1}_{N,m}}

\newcommand{\cC}{\mathcal{C}}
\newcommand{\cCi}{\cC_i}

\newcommand{\cmi}{c_{-i}}

\newcommand{\norm}[1]{\left \|#1 \right \|}
\newcommand{\Px}[1]{P_{\X}\left[#1 \right]}
\newcommand{\Pxi}[1]{P_{\X_i}\left[#1 \right]}

\newcommand{\Rd}{R_{d}}
\newcommand{\Rm}{R_{m}}

\newcommand{\lstar}{\lambda^\star}

\newcommand{\lt}{\lambda^\iter}
\newcommand{\ltp}{\lambda^{\iterp}}
\newcommand{\lit}{\lambda_i^\iter}
\newcommand{\ljt}{\lambda_j^\iter}
\newcommand{\litp}{\lambda_i^{\iterp}}

\newcommand{\bl}{\bar{\lambda}}

\newcommand{\blt}{\bl^{\iter}}
\newcommand{\bltp}{\bl^{\iterp}}
\newcommand{\pl}{\lambda_\perp}

\newcommand{\plt}{\pl^\iter}

\newcommand{\tpl}{\tilde{\lambda}_\perp}
\newcommand{\tplt}{\tpl^\iter}

\newcommand{\T}{^\top}
\newcommand{\map}[3]{#1: #2 \rightarrow #3}
\newcommand{\blkdiag}{\text{blkdiag}}
\newcommand{\until}[1]{\{1,\ldots,#1\}}
\newcommand{\col}{\textsc{col}}
\newcommand{\diag}{\text{diag}}

\newcommand*{\QEDB}{\hfill\ensuremath{\square}}%
\renewcommand{\qedsymbol}{$\blacksquare$}

\newcommand\oprocendsymbol{\hbox{$\square$}}
\newcommand\oprocend{\relax\ifmmode\else\unskip\hfill\fi\oprocendsymbol}
\def\eqoprocend{\tag*{$\square$}}
\def\er/{Erd\H{o}s-R\'enyi}

\graphicspath{{figs/}}

\newcommand{\nw}{n_w}

\def\algo/{{Primal TRADES}}
\def\algoc/{{Primal-Dual TRADES}}

\newtheorem{theorem}{Theorem}[section]

\newtheorem{definition}[theorem]{Definition} 
\newtheorem{lemma}[theorem]{Lemma}

\newtheorem{assumption}[theorem]{Assumption}
\newtheorem{remark}[theorem]{Remark} 

\newtheorem{standing}[theorem]{Standing Assumption}

\begin{document}
\title{Tracking-Based Distributed Equilibrium Seeking for Aggregative Games}
\author{Guido Carnevale, Filippo Fabiani, Filiberto Fele, Kostas Margellos, Giuseppe Notarstefano
\thanks{This result is part of the project ``Distributed Optimization for Cooperative Machine Learning in Complex Networks" (No PGR10067) that has received funding from the Ministero degli Affari Esteri e della Cooperazione Internazionale. F.~Fele gratefully acknowledges support from
	grant RYC2021-033960-I funded by MCIN/AEI/ 10.13039/501100011033 and European Union NextGenerationEU/PRTR, as well as from grant PID2022-142946NA-I00 funded by MCIN/AEI/ 10.13039/501100011033 and by ERDF A way of making Europe.}
\thanks{G.~Carnevale and G.~Notarstefano are with the Department of Electrical,  Electronic and Information Engineering,  Alma Mater Studiorum - Universita` di Bologna,  Bologna, Italy (\texttt{name.lastname@unibo.it}) F. Fabiani is with the IMT School for Advanced Studies Lucca, Piazza San Francesco 19, 55100 Lucca, Italy ({\tt filippo.fabiani@imtlucca.it}). F.~Fele is with the Department of Systems Engineering and Automation, University of Seville, Spain ({\tt ffele@us.es}). K. Margellos is with the Department of Engineering Science, University of Oxford, UK ({\tt kostas.margellos@eng.ox.ac.uk}).%
}
}

\maketitle

\begin{abstract}
    We propose fully-distributed algorithms for Nash equilibrium seeking in aggregative games over networks. We first consider the case where local constraints are present and we design an algorithm combining, for each agent, (i) the projected pseudo-gradient descent and (ii) a tracking mechanism to locally reconstruct the aggregative variable. To handle coupling constraints arising in generalized settings, we propose another distributed algorithm based on (i) a recently emerged augmented primal-dual scheme and (ii) two tracking mechanisms to reconstruct, for each agent, both the aggregative variable and the coupling constraint satisfaction. Leveraging tools from singular perturbations analysis, we prove linear convergence to the Nash equilibrium for both schemes. Finally, we run extensive numerical simulations to confirm the effectiveness of our methods and compare them with state-of-the-art distributed equilibrium-seeking algorithms.
\end{abstract}

\section{Introduction}
\label{sec:introduction}
Recent years have seen an increasing attention to the computation of (generalized) Nash equilibria in games over networks~\cite{menache2011network, scutari2014real,FacchineiKanzowGNE2010}. 
Indeed, numerous applications falling within different domains such as smart grids management~\cite{mohsenian2010autonomous,belgioioso2020energy}, economic market analysis~\cite{okuguchi2012theory}, cooperative control of robots~\cite{fabiani2018distributed}, electric vehicles charging~\cite{deori2018price,fele2020probably,cenedese2019charging}, network congestion control~\cite{barrera2014dynamic}, and synchronization of coupled oscillators in power grids~\cite{yin2011synchronization} can be modeled as networks of selfish agents -- aiming at optimizing their strategy according to an associated individual cost function -- that compete with each other over shared resources.

Among these examples, one can often find instances modeled as an \emph{aggregative} game, where the strategies of all the agents in the network are coupled through the so-called aggregative variable (expressing, e.g., the mean strategy), upon which each agent's cost function depends; see, e.g.,~\cite{jensen2010aggregative,parise2021analysis,belgioioso2022distributed} for a comprehensive overview.
Our work investigates such a framework proposing novel distributed algorithms for generalized Nash equilibrium (GNE) seeking under \emph{partial information}, i.e., assuming that each agent is only aware of its own local information (e.g., its strategy set and cost function) and can communicate only with few agents in the network. 
This restriction naturally calls for the design of fully-distributed mechanisms for GNE seeking.

Our approach is motivated by recent developments in cooperative optimization, where agents in a network collaborate to minimize the sum of individual objective functions depending both on local decision variables and an aggregative variable~\cite{li2021distributed,li2021distributedOnline,carnevale2022distributed,carnevale2022aggregative,wang2022distributed}.

\subsection{Related work}
In the context of NE problems in aggregative form, first attempts to design equilibrium seeking algorithms involve semi-decentralized approaches in which a central entity gathers and shares global quantities (such as the aggregative variable and/or a dual multiplier) with all the agents~\cite{grammatico2015decentralized,belgioioso2017semi,de2019continuous,grammatico2017dynamic,paccagnan2018nash,yi2019operator,kebriaei2021multipopulation,belgioioso2021semi}.

To relax the communication requirements,~\cite{koshal2016distributed} proposes a gradient-based algorithm for non-generalized games with diminishing step-size that relies on dynamic averaging consensus (see, e.g.,~\cite{zhu2010discrete,kia2019tutorial}) to reconstruct the aggregative variable in each agent. 
Such a method has been refined in~\cite{ye2021differentially} to deal with privacy issues and, as a consequence, only guarantees approximate equilibrium computations. 
In~\cite{parise2020distributed}, the distributed computation of an approximate Nash equilibrium is guaranteed through a best-response-based algorithm requiring multiple communication exchanges per iteration. 
In~\cite{cenedese2020asynchronous}, instead, an asynchronous distributed algorithm based on proximal dynamics is proposed.

Looking at GNE problems where the agents' strategies are coupled also by means of constraints,
in~\cite{parise2019distributed} the distributed computation of an approximate NE is guaranteed through an algorithm requiring, however, several communication exchanges per iteration. 
Exact convergence is instead guaranteed in~\cite{belgioioso2020distributed}, where a distributed algorithm with diminishing step-size is proposed, combining dynamic tracking mechanisms, monotone operator splitting, and the Krasnosel'skii-Mann fixed-point iteration. 
An exactly convergent distributed equilibrium-seeking algorithm with constant step-size is given in~\cite{gadjov2020single},
where the authors propose a distributed method based on a forward-backward splitting of two preconditioned operators requiring a double communication exchange per iteration. 
Finally, distributed equilibrium-seeking algorithms based on proximal best-responses are proposed in~\cite{bianchi2022fast}.

\subsection{Contributions}

The main contribution of the paper is the design and the analysis of novel, fully distributed iterative -- i.e., discrete-time -- algorithms for (generalized) NE seeking in aggregative games over networks. 
First, to address the case where local constraints are present, we combine a projected pseudo-gradient method with a local, auxiliary variable that compensates for the lack of knowledge of the aggregative variable in each agent. 
Successively, we deal with the case of coupling constraints, however, no local constraints are present. 
To achieve this, we take inspiration from a recent augmented primal-dual scheme for centralized, continuous-time optimization~\cite{qu2018exponential} and resort to (i) an averaging step to enforce consensus among the agents' multipliers and (ii) two auxiliary variables to locally reconstruct both the aggregative variable and the coupling constraint status. 
Both iterative schemes are analyzed from a system-theoretic perspective that allows us to establish linear convergence to the (G)NE. 
To the best of our knowledge, 
the algorithm proposed for the case where coupling constraints are present is the first distributed scheme in the literature with guaranteed linear convergence to a GNE (see Appendix~\ref{sec:linear} for the formal definition). As such, it constitutes the main contribution of our paper.
Moreover, as discussed in detail in Section~\ref{sec:SP_result}, such a linear convergence rate is enabled by our system-theoretic interpretation,
which offers a new proof-line perspective.
Further, we also guarantee linear convergence when only local constraints are present.
A similar result is also achieved by the recent contribution~\cite{bianchi2022fast} (see~\cite[Rem.~12]{bianchi2022fast}): our algorithm complements the proximal best-response scheme of \cite{bianchi2022fast} by constituting its gradient-based counterpart. Proximal algorithms require solving some optimization program (which in turn may rely on some iterative method), whereas our projected gradient descent step can allow a simpler update rule if the projection can be performed in an easy manner. As such, gradient-based approaches are often computationally less intensive compared to proximal ones, as verified in the numerical simulations of Section~\ref{sec:numerical_simulation} (see Table~\ref{table:execution_times}); however, such a conclusion is case-dependent.

As a side technical contribution, in contrast with existing methods, our algorithms (i) do not require compactness of the local feasible sets and (ii) allow for a general form of the aggregative variable, thus not necessarily requiring the mean of the agents' strategies to operate.
To better classify our work within the existing literature, Tables~\ref{table:unconstrained} and~\ref{table:constrained} compare it with the most relevant works. Specifically, Table~\ref{table:unconstrained} considers the framework without coupling constraints, while Table~\ref{table:constrained} the one with coupling constraints (GNE) (note that some of the technical conditions and variables appearing in the table entries will become clear in the sequel).

The analysis of our iterative algorithms is carried out by relying on a \emph{singular perturbations} approach that allows us to see each procedure as the interconnection between a \emph{slow} subsystem and a \emph{fast} one. 
Specifically, the slow dynamics are produced by the update of the strategies and, in the case with coupling constraints, of the mean of the multipliers over the network. 
The fast dynamics, instead, describes the evolution of the auxiliary variables used to compensate for the lack of knowledge of the global quantities and, in the case with coupling constraints, the consensus error among the agents' multipliers. 
Based on this interpretation, we construct two auxiliary, simplified subsystems, known as \emph{boundary layer} and \emph{reduced} system, to separately study the fast and slow dynamics, respectively. 
Leveraging this connection, we first provide the convergence properties of these auxiliary dynamics with Lyapunov-based arguments, and then we merge the obtained results to establish linear convergence to the (G)NE of the whole interconnection.
This last step relies on a general theorem (cf.~Theorem~\ref{th:theorem_generic}) considering a class of singularly perturbed systems that includes the proposed iterative algorithms. 
In detail, this theorem shows that global exponential stability results for the interconnection can be achieved, while typical results in literature only provide semi-global properties (see \cite[Prop.~8.1]{bof2018lyapunov}, or \cite[Ch.~11]{khalil2002nonlinear} for the continuous-time case).
To the best of our knowledge, similar results are not yet available in the literature: besides the construction of a novel, fully distributed iterative mechanism with appealing features for their practical implementation, they offer a new proof line for equilibrium-seeking problems.

Finally, we provide detailed numerical simulations to confirm the effectiveness of our methods and compare them with state-of-the-art distributed NE-seeking algorithms. %

\begin{table*}
    \centering
    \adjustbox{scale=0.83}{
\begin{tabular}{|c|c|c|c|c|c|}  
    \cline{2-6}
    \multicolumn{1}{c|}{} &\cite{koshal2016distributed} &\cite{parise2020distributed} &\cite{ye2021differentially} &\cite{bianchi2022fast} & Our algorithm\\ 
    \hline
    Linear rate &\xmark &\xmark &\xmark &\cmark &\cmark\\
    \hline
    Step-size &Diminishing &- &Diminishing &- &Constant\\
    \hline
    Exactness &\cmark &\xmark &\xmark &\cmark &\cmark\\
    \hline
    Communications per iterate &$1$ &$v$ &$1$ &$1$ &$1$\\
    \hline 
    Equilibrium assumptions &$F$ strictly monotone &\begin{tabular}{@{}c@{}}$\exists!$ equilibrium, \\ $\xibr(\cdot)$ non-expansive\end{tabular}  &$F$ strongly monotone &$F$ strongly monotone &$F$ strongly monotone\\
    \hline
    Local constraint set &Compact and convex &Compact and convex &Unconstrained &Closed and convex &Closed and convex\\
    \hline
    Gradient unboundedness &\cmark &\cmark &\xmark &\cmark &\cmark\\
    \hline
    Aggregative variable &$\tfrac{1}{N}\sum\limits_{i=1}^N \x_i$ &$\tfrac{1}{N}\sum\limits_{i=1}^N \x_i$ &$\tfrac{1}{N}\sum\limits_{i=1}^N \x_i$ &$\tfrac{1}{N}\sum\limits_{i=1}^N \x_i$ &$\tfrac{1}{N}\sum\limits_{i=1}^N \phi_i(\x_i)$\\
    \hline
    Algorithmic structure &Gradient-based &Best-response-based &Gradient-based &Proximal-based &Gradient-based\\
    \hline
    Graph &\begin{tabular}{@{}c@{}}Undirected, \\ time-varying\end{tabular} &Directed &\begin{tabular}{@{}c@{}}Undirected, \\ time-varying\end{tabular} &Undirected &Directed\\
    \hline
\end{tabular}
}
\caption{Setup without coupling constraints.}\label{table:unconstrained}
\end{table*}
\begin{table*}
    \centering
\adjustbox{scale=0.83}{
    \begin{tabular}{|c|c|c|c|c|c|}  
        \cline{2-6}
        \multicolumn{1}{c|}{} &\cite{parise2019distributed} &\cite{belgioioso2020distributed} &\cite{gadjov2020single} &\cite{bianchi2022fast} &Our algorithm\\ 
        \hline
        Linear rate &\xmark &\xmark &\xmark &\xmark &\cmark\\
        \hline
        Step-size &Constant &Diminishing &Constant &- &Constant\\
        \hline
        Exactness &\xmark &\cmark &\cmark &\cmark &\cmark\\
        \hline
        Communications per iterate &$v$ &$1$ &$2$ &$1$ &$1$\\
        \hline 
        Equilibrium assumptions &$F$ strongly monotone &$F$ cocoercive &$F$ strongly monotone &$F$ strongly monotone &$F$ strongly monotone\\
        \hline
        Local constraints &Compact and convex &Compact and convex &Compact and convex &Closed and convex &Unconstrained\\
        \hline
        Coupling constraints& Not specified &SCQ &SCQ &SCQ &$\kappa_1 I \leq AA\T \leq \kappa_2 I$\\
        \hline
        Gradient unboundedness& \xmark& \cmark &\cmark &\cmark &\cmark\\
        \hline
        Aggregative variable &$\tfrac{1}{N}\sum\limits_{i=1}^N \x_i$ &$\tfrac{1}{N}\sum\limits_{i=1}^N \x_i$ &$\tfrac{1}{N}\sum\limits_{i=1}^N \x_i$ &$\tfrac{1}{N}\sum\limits_{i=1}^N \x_i$ &$\tfrac{1}{N}\sum\limits_{i=1}^N \phi_i(\x_i)$\\
        \hline
        Algorithmic structure &Gradient-based &Gradient-based &Gradient-based &Proximal-based &Gradient-based\\
        \hline
        Graph &Directed &\begin{tabular}{@{}c@{}}Undirected, \\ time-varying\end{tabular} &Undirected &Undirected &Directed\\
        \hline
    \end{tabular}
}
    \caption{Setup with coupling constraints, where SCQ stands for Slater's Constraint Qualification.}\label{table:constrained}
    \end{table*}

\subsection{Paper organization}
In Section~\ref{sec:setup} we introduce aggregative games over networks, while 
in Section~\ref{sec:unconstrained} we propose and analyze a novel distributed algorithm to find NE when only local constraints are present.
In Section~\ref{sec:constrained} we devise a novel distributed GNE-seeking algorithm to address the case of linear coupling constraints. 
Finally, in Section~\ref{sec:numerical_simulation} we provide detailed numerical simulations to test our methods.
The proof of the result on singular perturbations -- instrumental in the derivation of our main theorems -- is deferred to Appendix~\ref{sec:auxiliary}; Appendices~\ref{sec:proof_SP} -- \ref{sec:proof_SP_primal_dual} gather the proofs of all other technical results and lemmas.

\smallskip

\textit{Notation}:
A matrix $M \in \R^{n\times n}$ is Schur if all its eigenvalues lie in the open unit disc.
The identity matrix in $\R^{m\times m}$ is $I_m$.
$0_m$ is the all-zero matrix in $\R^{m\times m}$. 
The vector of $N$ ones is denoted by $1_N$, while
$\oned \coloneqq 1_N \otimes I_d$ with $\otimes$ being the Kronecker product. 
Dimensions are omitted whenever clear from the context. 
Given a function of two variables $f: \R^{n_1} \times \R^{n_2} \to \R$, we denote as $\nabla_1 f \in \R^{n_1}$ its gradient with respect to its first argument and as $\nabla_2 f \in \R^{n_2}$ its gradient with respect to the second one. 
The vertical concatenation of column vectors $v_1, \dots, v_N \in \R^n$ is $\col (v_1, \dots,
v_N)$. 
$\R_{+}^n$ is the positive orthant in $\R^n$.
$\diag(v_1,\dots,v_n)$ denotes the diagonal matrix whose $i$-th diagonal element is given by $v_i$.
$\blkdiag(M_1, \dots, M_N)$ is the block diagonal matrix whose $i$-th block is $M_i \in \R^{n_i \times n_i}$.
Given a vector $x \in \R^{n}$ and a set $\X \subseteq \R^{n}$, $\Px{\x}$ denotes the projection of $\x$ on $\X$. 
For matrix (resp., vector) $A \in \R^{m \times n}$ ($v \in \R^n$), we denote as $[A]_j$ ($[v]_j$) its $j$-th row ($j$-th component). Given two matrices $A,B\in\mathbb{R}^{m\times m}$, $A\succ B$ (resp.~$A \succcurlyeq B$) is equivalent to saying that $A-B$ is positive definite (resp.~semidefinite).
Given $x \in \R^n$ and $M \in \R^{n \times n}$ such that $M=M\T \succ 0$, $\norm{x}_M = \sqrt{x\T M x}$. 

\section{Mathematical preliminaries}
\label{sec:setup}

\subsection{Problem definition and main assumptions}
We consider a population of $N \in \N$ agents -- designated by the set $\mc I \coloneqq \{1,\dots, N\}$ -- whose interaction is described by the following collection of coupled optimization problems:
\begin{subequations}\label{eq:problem_constrained}
\begin{numcases}{\forall i \in \mc I :}
		\underset{\x_i \in \X_i}{\min} \quad \Ji(\x_i, \sigma(\x)) & \\
		~\textrm{ s.t. } \quad A_i\x_i + \sum_{j \in \mc I \setminus \{i\}} A_j\x_j \leq \sum_{i\in\mc I}b_i. & \label{eq:coupl_constr}
\end{numcases}
\end{subequations}
In words, every agent $i\in\mc I$ seeks an individual strategy $\x_i \in \X_i\subseteq \R^{n_i}$ to minimize a local cost defined by the function $\map{\Ji}{\R^{n_i} \times \R^d}{\R}$, which depends on $\x_i$ as well as on some aggregate measure of other agents' strategies $\sigma(\x) \in \R^d$, where $x \coloneqq \col(\x_1,\dots,\x_N) \in \R^n$ and $n \coloneqq \sum_{i=1}^N n_i$. 
The agents' decisions shall satisfy some global constraints which can be expressed in the affine form $Ax\leq b$, where $A \coloneqq \left[ A_1 \, \cdots \, A_N\right] \in\R^{m\times n}$ and $b \coloneqq \sum_{i \in \mc I} b_i\in\R^m$. %
The aggregative variable $\sigma(\cdot)$ formally reads as 
\begin{align}\label{eq:sigma}
    \sigma(\x) \coloneqq \frac{1}{N} \sum_{i \in \mc I}\phii(\x_i),
\end{align}
where each \emph{aggregation rule} $\phii: \R^{n_i} \to \R^d$ models the contribution of the corresponding strategy $\x_i$ to the aggregate $\sigma(\x)$. 
We define the constraint functions $c_i: \R^{n_i} \to \R^m$, $\cmi: \R^{n-n_i} \to \R^{m}$, and $c :\R^n \to \R^m$ as follows:
\begin{subequations}\label{eq:c}
    \begin{align}
        c_i(\x_i) &= A_i\x_i - b_i
        \\
        \cmi(\xmi) &= \sum_{j \in \mc I \setminus \{i\}}(A_j\x_j - b_j)
        \\
        c(\x) &= c_i(\x_i) + \cmi(\xmi) = A\x - b,
    \end{align}
\end{subequations} 
where $\xmi \coloneqq \col(\x_1,\dots,x_{i-1},x_{i+1},\dots,\x_{N}) \in \R^{n-n_i}$.
Then, the collective vector of strategies $\x$ belongs to the feasible set $\cC \coloneqq \{\x \in \X \mid c(\x) \leq 0\} \subseteq \R^n$.

We refer to any equilibrium solution to the collection of inter-dependent optimization problems~\eqref{eq:problem_constrained} as aggregative GNE~\cite{FacchineiKanzowGNE2010} (or simply GNE), and to the problem of finding such an equilibrium as GNE problem (GNEP) in aggregative form -- as opposed to a NE problem (NEP) which is characterized by local constraints only.
We will design distributed algorithms to find aggregative GNEs, which formally correspond to the following definition:
\begin{definition}[\textup{Generalized Nash equilibrium} \cite{FacchineiKanzowGNE2010}]\label{def:GNE}
	A collective vector of strategies $\xstar \in \cC$ is a GNE of \eqref{eq:problem_constrained} if we have:
    \begin{equation*}
		\Ji(\xistar, \sigma(\xstar)) \leq \underset{\x_i\in\cCi(\xstari)}{\min} \, \Ji(\x_i,\tfrac{1}{N}\phii(\x_i) + \sigmai(\xstari)), %
	\end{equation*}
	for all $i \in \mc I$, with $\cCi(\xmi) \coloneqq \{\x_i\in \X_i \mid A_i\x_i \leq b_i - \cmi(\xmi)\}$ and $\sigmai(\xstari) := \frac{1}{N}\sum_{j \in \mc I \setminus \{i\}}\phij(\xstar_j)$.\oprocend
\end{definition} 
We remark that the definition of NE follows directly from the above by noting that, in the case without coupling constraints, it holds $\cCi(\xstari) \equiv \X_i$. %

An equivalent definition of GNE requires one to find a fixed point of the \emph{best response} mapping $\xibr : \R^{n - n_i} \to \R^{n_i}$ of each agent, which is formally defined as:
\begin{align*}%
    \xibr(\xmi) &\in \argmin_{\x_i \in \cCi(\xmi)} \:  \: \Ji\left(\x_i, \sigma(\x)\right) 
    \notag\\
    &= \argmin_{\x_i \in \cCi(\xmi)} \:  \: \Ji\left(\x_i, \tfrac{1}{N}\phii(\x_i) + \sigmai(\xmi)\right),
\end{align*}
In fact, a collective vector of strategies $\xstar$ is a GNE if, for all $i \in \mc I$, $\x_i^\star = \xibr(\xmi^\star)$.
Next, we enforce customary assumptions about the regularity of some quantities in~\eqref{eq:problem_constrained}. 
\begin{standing}[Cost functions]\label{ass:costs}
	For all $i \in \mc I$, the function $J_i(\cdot,\phii(\cdot)/N + \sigmai(\xmi))$ is of class $\cC^1$, i.e., its derivative exists and is continuous, for all $\xmi \in \R^{n-n_i}$. \QEDB%
\end{standing}
A key ingredient in this game-theoretic framework is the so-called \emph{pseudo-gradient mapping} $F: \R^n \to \R^n$:
\begin{equation}\label{eq:F_definition}
	\hspace{-.2cm}F(\x) \coloneqq \col(\nabla_{\x_1} J_1(\x_1,\sigma(\x)), \dots, \nabla_{\x_N} J_N(\x_N,\sigma(\x))).
\end{equation}
With this regard, we also make the following assumption.
\begin{standing}[Strong monotonicity and Lipschitz continuity]
	\label{ass:objective_function}
	$F$ is $\mu$-strongly monotone, i.e., there exists $\mu > 0$ such that
	\begin{equation*}
		(F(x) - F(y))\T (\x - \y) \ge \mu\norm{\x - \y}^2,
	\end{equation*}
	for any $\x, \y \in \R^{n}$.
	Moreover, given any $\x_i, \x_i^\prime \in \R^{n_i}$ and $y, y^\prime \in \R^{n-n_i}$, for all $i \in \mc I$, we assume that 
	\begin{align*}
		\|\nabla_{\x_i}\Ji(\x_i,\phii(\x_i)/N + y) -  \nabla_{\x_i^\prime}\Ji(\x_i^\prime,\phii(\x_i^\prime)/N + y^\prime)\|
        & \leq \beta_1\norm{\col(\x_i,y) - \col(\x_i^\prime,y^\prime)}
        \\
		\norm{\nabla_1\Ji(\x_i,y) - \nabla_1\Ji(\x_i^\prime,y^\prime)} & \leq \beta_1\norm{\col(\x_i,y) - \col(\x_i^\prime,y^\prime)}
		\\
		\norm{\nabla_2\Ji(\x_i,y) - \nabla_2\Ji(\x_i^\prime,y^\prime)} & \leq \beta_2\norm{\col(\x_i,y) - \col(\x_i^\prime,y^\prime)}
		\\
		\norm{\phii(\x_i) - \phii(\x_i^\prime)} & \leq \beta_3\norm{\x_i - \x_i^\prime}. \eqoprocend
	\end{align*}  %
\end{standing}
While assumptions on strong monotonicity and Lipschitz continuity of the game mapping are quite standard in the literature~\cite{paccagnan2018nash,yi2019operator,belgioioso2022distributed}, the second part of Standing Assumption~\ref{ass:objective_function} specializes the Lipschitz properties of the gradients of the cost functions in both the local and aggregate variables, as well as of each single aggregation rule $\phii(\cdot)$. %
Note that we assume partial information, i.e., each agent $i$ is only aware of its own local information $\x_i$, $\Ji$, $\phii$, $A_i$, and $b_i$. 
Moreover, each agent can exchange information with a subset of $\mc I$ only. Specifically,
we consider a network of agents whose communication is performed according to a directed graph $\cG = (\mc I, \cE)$, with $\cE \subset \mc I^2$ such that $i$ can receive information from agent $j$ only if the edge $(j,i)\in\cE$.  
The set of in-neighbors of $i$ is represented by $\cN_i \coloneqq \{j \in \mc I \mid (j,i) \in \cE\}$ (where also $i\in\cN_i$), while $\cN_i^{\text{out}} \coloneqq \{j \in \mc I \mid (i,j) \in \cE\}$ denotes the set of out-neighbors of the agent $i$. 
Graph $\cG$ is associated with a weighted adjacency matrix $\cW \in\R^{N\times N}$ whose entries satisfy $w_{ij} >0$ whenever $(j,i)\in \cE$ and $w_{ij} =0$ otherwise. 
The next assumption characterizes the considered graphs.
\begin{standing}[Network]
	\label{ass:network}
	The graph $\cG$ is strongly connected, i.e., for every pair of nodes $(i,j) \in \mc I^2$ there exists a path of directed edges 
	that goes from $i$ to $j$, and the matrix $\cW$ is doubly stochastic, namely it holds that:
    \begin{align*}
        \cW 1_N = 1_N, \quad 1_N\T \cW = 1_N\T.\eqoprocend
    \end{align*}
\end{standing}

\subsection{A key result on singularly perturbed systems}
\label{sec:SP_result}

The convergence analysis of the iterative schemes introduced in Section~\ref{sec:unconstrained_algorithm} and~\ref{sec:constrained} exploits a system-theoretic perspective based on \emph{singular perturbation}, that strongly relies on the following crucial result proved in Appendix~\ref{sec:auxiliary}.
\begin{theorem}[Global exponential stability for singularly perturbed systems]\label{th:theorem_generic}
	Consider the system
	\begin{subequations}\label{eq:interconnected_system_generic}
		\begin{align}
			\xtp &= \xt + \delta f(\xt,\wt)\label{eq:slow_system_generic}
			\\
			\wtp &= g(\wt,\xt,\delta),\label{eq:fast_system_generic}
		\end{align}
	\end{subequations}	
	with $\xt \in \mathcal{D} \subseteq \R^n$, $\wt \in \R^m$, $\map{f}{\mathcal{D} \times \R^m}{\R^n}$, $\map{g}{\R^m \times \R^n \times\R}{\R^m}$, $\delta > 0$. 
    Let $f$ and $g$ be Lipschitz continuous with respect to both $\xt$ and $\wt$, with Lipschitz constants $L_f$ and $L_g$, respectively.
    Assume that there exists an $L_h$-Lipschitz continuous function $h: \R^n \to \R^m$ such that, for any $\x \in \R^n$, %
    \begin{displaymath}
    	h(\x) = g(h(\x),\x,\delta),
    \end{displaymath}
	and further assume that there exists $\xstar \in \R^n$ such that 
	\begin{displaymath}
		0 = \delta f(\xstar, h(\xstar)).
	\end{displaymath}
    Then, let
	\begin{equation}\label{eq:reduced_system_generic}
		\xtp = \xt + \delta  f(\xt,h(\xt))
	\end{equation}
	be the \emph{reduced system} and
	\begin{equation}\label{eq:boundary_layer_system_generic}
		\psitp = g( \psit + h(\x),\x,\delta) - h(\x)
	\end{equation}
	be the \emph{boundary layer system} with $ \psit \in \R^m$.
	
	Assume that there exists a continuous function $U: \R^m \to \R$ and $\bar{\delta}_1 > 0$ such that, for any $\delta \in (0,\bar{\delta}_1)$ (cf.~\eqref{eq:interconnected_system_generic}), there exist $b_1, b_2, b_3, b_4 > 0$ such that for any $ \psi,  \psi_1,  \psi_2 \in \R^m$, $\x \in \R^n$,
	\begin{subequations}\label{eq:U_generic}
		\begin{align}
			b_1 \norm{ \psi}^2 \leq U( \psi) &\leq b_2\norm{ \psi}^2 \label{eq:U_first_bound_generic}
			\\
			U(g( \psi  +  h(\x),\x,\delta) - h(\x)) -  U( \psi)  &\leq  -b_3\norm{ \psi}^2\label{eq:U_minus_generic}
			\\
			|U( \psi_1)-U( \psi_2)|&\leq b_4\norm{\psi_1- \psi_2}\norm{\psi_1}
			\notag\\
			&\hspace{.5cm}
			+ b_4\norm{\psi_1- \psi_2}\norm{\psi_2}.\label{eq:U_bound_generic}
		\end{align}
	\end{subequations}
	Further, assume there exists a continuous function $W:\mathcal{D} \to \R$ and $\bar{\delta}_2 > 0$ such that, for any $\delta \in (0,\bar{\delta}_2)$, there exist $c_1, c_2, c_3, c_4 > 0$ such that for any 
	$\x, \x_1, \x_2, \x_3 \in \mathcal{D}$
	\begin{subequations}\label{eq:W_generic}
		\begin{align}
			c_1 \norm{\x - \xstar}^2 \leq W(\x) &\leq c_2\norm{\x - \xstar}^2\label{eq:W_first_bound_generic}
			\\
			W(\x + \delta  f(\x,h(\x)))  -  W(\x)  &\leq  -\delta c_3\norm{\x - \xstar}^2\label{eq:W_minus_generic}
			\\
			|W(\x_1)-W(\x_2)|&\leq c_4\norm{\x_1-\x_2}\norm{\x_1 - \xstar}
			+ c_4\norm{\x_1-\x_2}\norm{\x_2 - \xstar}.\label{eq:W_bound_generic}
		\end{align}
	\end{subequations}

	Then, there exist $\bar{\delta} \in (0,\min\{\bar{\delta}_1,\bar{\delta}_2\})$, $a_1 >0$, and $a_2 > 0$ such that, for all $\delta \in (0,\bar{\delta})$, it holds
	\begin{align*}
		\norm{\begin{bmatrix}
				\xt - \xstar\\
				\wt - h(\xt)
		\end{bmatrix}} \leq a_1\norm{\begin{bmatrix}
				\x^0 - \xstar\\
				\w^0 - h(\x^0)
		\end{bmatrix}} e^{-a_2t},
	\end{align*}
	for all $(\x^0,\w^0) \in \mathcal{D} \times \R^m$.\QEDB
\end{theorem}
Theorem~\ref{th:theorem_generic} establishes a stability result for the system in~\eqref{eq:interconnected_system_generic}, that can be thought of as an interconnection of a fast and a slow subsystem (for sufficiently small $\delta>0$). This is schematically illustrated in Fig.~\ref{fig:bd_interconnection}. 
\begin{figure}[H]
	\centering
	\includegraphics[scale=.95]{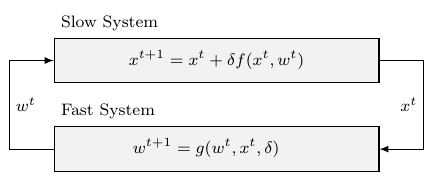}
	\caption{Block diagram of the original interconnected system~\eqref{eq:interconnected_system_generic}.}
	\label{fig:bd_interconnection}
\end{figure}
To analyze this interconnection we study separately the simplified auxiliary systems~\eqref{eq:reduced_system_generic}-\eqref{eq:boundary_layer_system_generic}.
For any $x \in \R^{n}$, $h(x)$ is a parametric equilibrium of the fast subsystem; we can then fix the slow state $x^t = x$ into the fast dynamics~\eqref{eq:fast_system_generic}, to obtain the so-called boundary layer system, as pictorially shown in Fig.~\ref{fig:bd_boundary_layer}. 
This is the auxiliary system described by~\eqref{eq:boundary_layer_system_generic}, whose state $\psi^t$ encodes the distance of the state $\wt$ of the fast subsystem from the equilibrium $h(x)$, once $x$ is fixed.
Existence of a Lyapunov-like function with properties as in \eqref{eq:U_generic} ensures then that for any $x \in \R^{n}$, the boundary layer system is exponentially stable.
\begin{figure}[H]
	\centering
	\includegraphics[scale=.9]{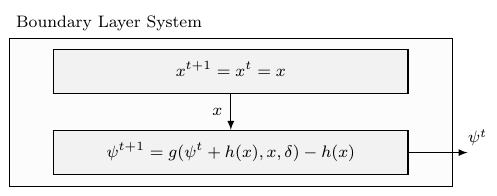}
	\caption{Block diagram of the boundary layer system~\eqref{eq:boundary_layer_system_generic}.}
	\label{fig:bd_boundary_layer}
\end{figure}
Setting now $w^t = h(x^t)$ for all $t\geq0$ in \eqref{eq:slow_system_generic}, i.e., considering the fast state at its parametric equilibrium, we obtain the so-called reduced system, i.e., the auxiliary system~\eqref{eq:reduced_system_generic}.
This is schematically shown in Fig.~\ref{fig:bd_reduced_system}.
Existence of a Lyapunov-like function with properties as in~\eqref{eq:W_generic} ensures then that $x^\star$ is globally exponentially stable for the reduced system.
By properly combining these two Lyapunov functions, Theorem~\ref{th:theorem_generic} ensures that, for sufficiently small values of $\delta$, the point $(\xstar,h(\xstar))$ is globally exponentially stable for the original interconnected system~\eqref{eq:interconnected_system_generic}. The detailed proof is provided in Appendix~\ref{sec:auxiliary}.
\begin{figure}[H]
	\centering
	\includegraphics[scale=.9]{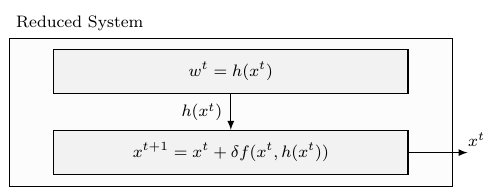}
	\caption{Block diagram of the reduced system~\eqref{eq:reduced_system_generic}.}
	\label{fig:bd_reduced_system}
\end{figure}

We will show next how our algorithms, namely \algo/ and \algoc/, can be recast in the form of the interconnected system \eqref{eq:interconnected_system_generic} while satisfying all assumptions of Theorem \ref{th:theorem_generic}, and hence prove their convergence in Theorems~\ref{th:convergence} and~\ref{th:convergence_primal_dual}, respectively, provided in the next sections.
Compared with traditional approaches, taking a singular perturbation view offers a novel proof line for (generalized) equilibrium-seeking problems.
\section{Aggregative games over networks\\ without coupling constraints}\label{sec:unconstrained_algorithm}

\subsection{\algo/}
\label{sec:unconstrained}

In this section, we introduce and analyze Primal TRacking-based Aggregative Distributed Equilibrium Seeking (TRADES), a fully-distributed iterative NE seeking algorithm for a special case of the aggregative game described by~\eqref{eq:problem_constrained}, i.e., where the local decision spaces are decoupled. 
Formally,
\begin{equation}\label{eq:problem}
	\forall i \in \mc I : \underset{\x_i \in \X_i}{\min} \; \Ji(\x_i, \sigma(\x)),
\end{equation}
where $\X_i \subseteq \R^{n_i}$, the local feasible set known to agent $i$ only, satisfies the following conditions:
\begin{assumption}\label{ass:local_feasible_set}
	For all $i \in \mc I$, the feasible set $\X_i$ is nonempty, closed, and convex.\QEDB
\end{assumption}
\begin{remark}
		The general structure of the aggregative variable $\sigma(x)$ in \eqref{eq:sigma} can accommodate ``soft'', possibly nonlinear, coupling constraints; these can be incorporated in the game~\eqref{eq:problem} by penalizing their residual in the players' cost functions.
        \oprocend
\end{remark}

Let $\xit \in \R^{n_i}$ be the strategy chosen by each agent $i$ at iteration $\iter \ge 0$. 
Taking its convex combination with a projected pseudo-gradient step may be an effective way to steer each agent's strategy to the best response $\xibr(\sigmai(\xmit))$. 
When applied to problem~\eqref{eq:problem}, it reads as
\begin{equation}\label{eq:desired_update}
    \xitp = \xit + \delta\left(\Pxi{\xit - \gamma \nabla_{\x_i}\Ji(\xit, \sigma(\xt))}-\xit\right),
\end{equation}
where $\delta \in (0,1)$ is a constant performing the combination and $\gamma > 0$ plays the role of the gradient step-size.
 We point out that the chain rule and the definition of $\sigma(\xt)$ (cf.~\eqref{eq:sigma}) lead to $\nabla_{\x_i}\Ji(\xit, \sigma(\xt)) = \nabla_1 J_i(\xit,\sigma(\xt)) + \frac{\nabla\phii(\xit)}{N}\nabla_2 \Ji(\xit,\sigma(\xt))$.
In our distributed setting, however, agent $i$ cannot access the global aggregate variable $\sigma(\xt)$. 
 To compensate for this lack of information, we rely on the locally available $\phii(\xit)$ and the auxiliary variable $\zit \in \R^{d}$. 
 Thus, for all $i \in \mc I$, let the operator $\tFi: \R^{n_i} \times \R^d \to \R^{n_i}$ be defined as
 \begin{equation*}
    \tFi(\x_i,s) \coloneqq \nabla_1 J_i(\x_i,s) + \frac{\nabla\phii(\x_i)}{N}\nabla_2 \Ji(x_i,s),
 \end{equation*}
 and, in accordance, we modify the update~\eqref{eq:desired_update} as 
\begin{equation}
    \xitp  =  \xit  +  \delta\left(P_{\X_i}\left[\xit - \gamma \tFi\left(\xit,\phii(\xit) + \zit\right)\right]-\xit\right),\label{eq:implementable_update}
\end{equation}
which can be directly implemented without violating the distributed nature of the algorithm. 
By comparing~\eqref{eq:desired_update} and~\eqref{eq:implementable_update}, we note that the global term $\sigma(\xt)$ has been replaced by the locally available proxy $\phii(\xit) + \zit$.
Therefore, if
\begin{equation}
    \zit \rightarrow -\phii(\xit) + \sigma(\xt), %
\end{equation}
then the implementable law~\eqref{eq:implementable_update} coincides with the desired one given in~\eqref{eq:desired_update}. 
Note that $\zit$ encodes the estimate of $\sigma(\xit) - \phii(\xit)$, i.e., the aggregate of all other agents' strategies except for the $i$-th one. %
For this reason, we update each auxiliary variable $\zit$ according to the following causal version of the perturbed average consensus scheme (see, e.g.,~\cite{carnevale2022nonconvex}, where a similar scheme has been used to locally compensate the missing knowledge of the global gradient of a distributed consensus optimization problem):
\begin{equation}
    \zitp = \sum_{j\in\cN_i}w_{ij}\zjt + \sum_{j\in\cN_i}w_{ij}\phij(\xjt) - \phii(\xit).\label{eq:z_local_update}
\end{equation}
This is implementable in a fully distributed fashion since it only requires communication with neighboring agents $j \in \cN_i$. 
We report the whole algorithmic structure in Algorithm~\ref{algo:unconstrained} and, from now on, we will refer to it as \algo/.
\begin{algorithm}[t!]
	\begin{algorithmic}
		\State \textbf{Initialization}: $\x_i^0 \in \X_i, \z_i^0 = 0$.
		\For{$t=1, 2, \dots$}
        \begin{subequations}\label{eq:local_update}
            \begin{align}
                \xitp &= \xit + \delta\left(\Pxi{\xit - \gamma\tFi\left(\xit,\phi_i(\xit) + \zit\right)} -\xit\right)
                \\
                \zitp &= \sum_{j\in\cN_i}w_{ij}\zjt + \sum_{j\in\cN_i}w_{ij}\phij(\xjt) - \phii(\xit).\label{eq:z_local}
            \end{align}
        \end{subequations}
		\EndFor
	\end{algorithmic}
	\caption{\algo/ (Agent $i$)}
	\label{algo:unconstrained}
\end{algorithm}
We note that Algorithm~\ref{algo:unconstrained} requires the initialization $\z_i^0 = 0$ for all $i \in \mc I$; we will discuss in the sequel the interpretation of this particular initialization. The local update~\eqref{eq:local_update} leads to the stacked vector form of \algo/, namely
\begin{subequations}\label{eq:global_update}
    \begin{align}
        \xtp &= \xt + \delta\big(P_{\X}\big[\xt - \gamma\tF\left(\xt,\phi(\xt) + \zt\right)\big] - \xt\big)
        \\
        \ztp &= \Wd\zt + (\Wd - I)\phi(\xt),
  \end{align}  
\end{subequations}
with $\X \coloneqq \prod_{i \in \mc I} \X_i \subseteq \R^n$, $\Wd \coloneqq \cW \otimes I_d \in \R^{Nd \times Nd}$, $\zt \coloneqq \col(\z_{1,\iter},\dots,\z_{N,\iter})$, $\phi(\xt) \coloneqq \col(\phi_1(\x_1^\iter),\dots,\phi_N(\x_N^\iter))$, and $\tF(\xt,\phi(\xt) + \zt) \coloneqq \col(\tilde{F}_1(\x_{1}^{\iter},\phi_1(\x_{1}^{\iter}) + z_{1}^{\iter}), \dots,
\tilde{F}_N(\x_{N}^{\iter},\phi_N(\x_{N}^{\iter}) + z_{N}^{\iter}))$.
We remark that, since $\F$ is $\mu$-strongly monotone (cf. Assumption~\ref{ass:objective_function}) and $\X$ nonempty, closed, and convex (cf. Assumption~\ref{ass:local_feasible_set}), there exists a unique Nash equilibrium $\xstar \in \R^{n}$ for~\eqref{eq:problem}. 
Moreover, for such an equilibrium it holds 
\begin{align*}
    \xstar = \Px{\xstar - \gamma \F(\xstar)},
\end{align*} 
for any $\gamma > 0$, see~\cite[Ch.~12]{facchinei2003finite}.
This result, in turn, guarantees that $\xstar =\xstar + \delta(\Px{\xstar - \gamma \F(\xstar)} - \xstar)$ for any $\delta > 0$.
We establish next the properties of \algo/ in computing the NE $\xstar$ of problem~\eqref{eq:problem}.
\begin{theorem}\label{th:convergence}
	Consider the dynamics in \eqref{eq:global_update} and Assumption~\ref{ass:local_feasible_set}.
	There exist constants $\bar{\delta}, \bar{\gamma}, a_1, a_2 > 0$ such that, for any $\delta \in (0,\bar{\delta})$, $\gamma \in (0,\bar{\gamma})$ and $(\x^0,\z^0) \in \R^{n+Nd}$ such that $\oned\T \z^0 = 0$, it holds 
	\begin{align*}
		\norm{\xt - \xstar} \leq a_1 e^{-a_2\iter}.\eqoprocend
	\end{align*}
\end{theorem}
The proof of Theorem~\ref{th:convergence} relies on a \emph{singular perturbation} analysis of system~\eqref{eq:global_update}, and will be given in the next subsection.
\subsection{Proof of Theorem \ref{th:convergence}}
\label{sec:SP}

We build the framework to prove Theorem~\ref{th:convergence} by analyzing \eqref{eq:global_update} under a singular perturbations lens. 
We therefore establish the related proof in five steps:
\emph{1. Bringing \eqref{eq:global_update} in the form of \eqref{eq:interconnected_system_generic}:} we leverage the initialization $\z^0$ so that $\oned\T\z^0 = 0$ to introduce coordinates $\bz \in \R^d$ and $\pz \in \R^{(N-1)d}$ defined as:
    \begin{align}\label{eq:change_mean}
        \begin{bmatrix}
            \bz
            \\
            \pz
        \end{bmatrix} \coloneqq \begin{bmatrix}
            \frac{\oned\T}{N}\\
            \Rd\T
        \end{bmatrix}\z \implies \z = \oned\bz + \Rd\pz,
    \end{align}
    where $\Rd \in \R^{Nd \times (N-1)d}$ with $\norm{\Rd} = 1$ is such that
    \begin{align}
    \Rd \Rd\T = I- \frac{\oned\oned\T}{N}~\text{ and }~ \Rd\T\oned = 0. \label{eq:eq_Rd}
    \end{align}
    Then, by using the definition of $\bz$ given in~\eqref{eq:change_mean}, the associated dynamics reads as
    \begin{align}
        \bztp &= \frac{\oned\T}{N}\ztp
        \stackrel{(a)}{=}
        \frac{\oned\T}{N}\Wd\zt + \frac{\oned\T}{N}(\Wd - I)\phi(\xt) 
        \notag\\
        &
        \stackrel{(b)}{=} 
        \frac{\oned\T}{N}\zt
        \stackrel{(c)}{=} 
        \frac{\oned\T}{N}\left(\oned \bzt + \Rd\pzt\right) 
        \stackrel{(d)}{=} \bzt,\label{eq:barz_update}
    \end{align} 
    where in $(a)$ we exploit the update~\eqref{eq:global_update}, in $(b)$ we use the facts that, in view of Standing Assumption~\ref{ass:network}, (i) $\oned\T \Wd = \oned\T$ and (ii) $\oned\T (\Wd - I) = 0$, in $(c)$ we rewrite $\zt$ according to~\eqref{eq:change_mean}, and in $(d)$ we use the fact that $\oned\T \Rd = 0$.
    Thus,~\eqref{eq:barz_update} leads to $\bztp \equiv \bz^0 \equiv 0$ for all $\iter \ge 0$, where the last equality follows by the initialization $\oned\T z^0 = 0$ and the definition of $\bz$ (cf.~\eqref{eq:change_mean}). 
    We are thus entitled to ignore the null dynamics of $\bzt$ and, according to~\eqref{eq:change_mean}, we equivalently rewrite~\eqref{eq:global_update} as
    \begin{subequations}\label{eq:global_update_mean}
        \begin{align}
            \xtp &= \xt  +  \delta\big(P_{\X}[\xt - \gamma \tF(\xt,\phi(\xt) + \Rd\pzt)]-\xt \big)\label{eq:x_update_mean}
            \\
            \pztp &= \Rd\T \Wd\Rd\pzt + \Rd\T(\Wd - I)\phi(\xt).\label{eq:z_update_mean}
      \end{align}  
    \end{subequations}
    For any $\iter \geq0$, the interconnected system~\eqref{eq:global_update_mean} can thus be obtained from \eqref{eq:interconnected_system_generic} by setting 
    \begin{align}\label{eq:assign}
        \begin{split}
            w^t &:= \pzt
            \\
            f(\xt,\wt) &:= \Px{\xt - \gamma \tF(\xt,\phi(\xt) + \Rd\wt)}-\xt 
            \\
            g(\wt,\xt) &:= \Rd\T \Wd\Rd\wt + \Rd\T(\Wd - I)\phi(\xt). 
        \end{split}
    \end{align}
    In particular, we refer to the subsystem~\eqref{eq:x_update_mean} as the slow system, while we refer to~\eqref{eq:z_update_mean} as the fast one. 

    \emph{2. Equilibrium function $h~$:} under the expression for $\Rd\Rd\T$ in~\eqref{eq:eq_Rd} and since $\cW$ is doubly stochastic (cf. Standing Assumption~\ref{ass:network}) notice that for any $\xt = \x \in \R^n$,
    \begin{equation}\label{eq:man_definition}
        \pz = h(\x) \coloneqq -\Rd\T\phi(\x)
    \end{equation}
    constitutes an equilibrium of~\eqref{eq:z_update_mean}. %
    Since $\Rd\T \Wd\Rd$ is Schur in view of Standing Assumption~\ref{ass:network}, we interpret~\eqref{eq:z_update_mean} as a strictly stable linear system with nonlinear input $\Rd\T(\Wd - I)\phi(\xt)$ parametrizing the equilibrium of the subsystem. 
    The role of $\gamma$ is to slow down the variation of $\xt$ so that $\pz$ always remains close to the parametrized equilibrium $h(\xt)$.

    \emph{3. Boundary layer system and satisfaction of \eqref{eq:U_generic}:} 
    the so-called boundary layer system associated to~\eqref{eq:global_update_mean} can be constructed by fixing $\xt = x$ for all $t\geq0$, for some arbitrary $x \in \mathbb{R}^n$ in \eqref{eq:z_update_mean}, and rewriting it according to the error coordinates $\tzt := \pzt - h(\x)$. 
    Using \eqref{eq:eq_Rd}, we obtain that
    \begin{equation}\label{eq:bl}
        \tztp = \Rd\T \Wd\Rd\tzt.
    \end{equation}
    Notice that the latter is in the form of \eqref{eq:boundary_layer_system_generic} with $\psi =\tzt$, and $g(\psi  +  h(\x),\x) - h(\x) = \Rd\T \Wd\Rd\tzt$.
    The next lemma provides a Lyapunov function for \eqref{eq:bl}.
    \begin{lemma}\label{lemma:bl}
        Consider system~\eqref{eq:bl}. Then, there exists a continuous function $U :\R^{(N-1)d} \to \R$ satisfying \eqref{eq:U_generic} with $\tz$ in place of $\psi$.\QEDB
    \end{lemma}
    \emph{4. Reduced system and satisfaction of \eqref{eq:W_generic}:}
   the so-called reduced system can be obtained by plugging into \eqref{eq:x_update_mean} the fast state at its steady state equilibrium, i.e., we consider $\zt = h(\xt)$ for any $t \ge 0$. We thus have
   \begin{equation}\label{eq:rs}
        \xtp = \xt + \delta\left(\Px{\xt - \gamma \tF(\xt,\phi(\xt) + \Rd h(\xt))}-\xt\right).
   \end{equation}
    Due to~\eqref{eq:eq_Rd} we have that $\tF(\xt,\phi(\xt) + \Rd h(\xt)) = \tF(\xt,\oned\sigma(\xt)) = F(\xt)$, so~\eqref{eq:rs} is equivalent to 
    \begin{equation}\label{eq:rs_explicit}
        \xtp = \xt + \delta\left(\Px{\xt - \gamma F(\xt)}-\xt\right).
    \end{equation}
    The next lemma provides a Lyapunov function for \eqref{eq:rs}.%
    \begin{lemma}\label{lemma:rs}
        Consider system~\eqref{eq:rs}. 
        Let $\xstar \in \mathbb{R}^n$ be such that $f(\xstar,h(\xstar)) = 0$ with $f$ defined as in~\eqref{eq:assign}.
        Then, there exist a continuous function $W: \R^n \to \R$ and $\bar{\gamma}, \bar{\delta}_2 >0$ such that, for any $\gamma \in (0,\bar{\gamma})$ and $\delta \in (0,\bar{\delta}_2)$, $W$ satisfies~\eqref{eq:W_generic}. \QEDB
    \end{lemma}
\emph{5. Lipschitz continuity of $f$, $g$ and $h$:}
as we will be invoking Theorem~\ref{th:theorem_generic}, we need to ensure that the Lipschitz continuity assumptions required by the theorem are satisfied. In particular, we require $f$ and $g$ in \eqref{eq:assign} to be Lipschitz continuous with respect to both arguments and $h$ in \eqref{eq:man_definition} to be Lipschitz continuous with respect to $x$.

Lipschitz continuity of $f$ follows by the fact that $\nabla J_i$ is Lipschitz continuous due to Standing Assumption~\ref{ass:objective_function}. To show Lipschitz continuity of $g$ in \eqref{eq:assign} notice that for any $\w, \w^\prime \in \R^{(N-1)d}$ and any $\x, \x^\prime \in \R^{n}$,
\begin{align*}
    &\norm{\Rd\T \Wd\Rd (\w - \w^\prime) + \Rd\T(\Wd - I)(\phi(\x) -\phi(\x^\prime))} 
     \notag\\
    &\leq \norm{\Rd\T \Wd\Rd}\norm{\w - \w^\prime} 
    + \beta_3\norm{\Rd\T(\Wd - I)}\norm{\x - \x^\prime},%
\end{align*}
where the inequality is due to triangle inequality and the fact that by Standing Assumption~\ref{ass:objective_function}, $\phi$ is Lipschitz continuous with Lipschitz constant $\beta_3$. To show Lipschitz continuity of $h$, notice that for any $\x, \x^\prime \in \R^n$,
    \begin{equation*}
    \norm{h(\x) - h(\x^\prime)} \leq \beta_3 \|\Rd\| \norm{\x - \x^\prime} = \beta_3 \norm{\x - \x^\prime},%
\end{equation*}
where the inequality follows from \eqref{eq:man_definition} and Lipschitz continuity of $\phi$, while the equality from the fact that $\|\Rd\| = 1$.

    By combining Lemma~\ref{lemma:bl} and \ref{lemma:rs} with the Lipschitz conditions expressed above, Theorem~\ref{th:theorem_generic} can therefore be applied. 
Thus, there exists $\bar{\delta} \in (0,\bar{\delta}_2)$ so that $(\xstar,h(\xstar))$ is an exponentially stable equilibrium for~\eqref{eq:global_update_mean}.

\section{Generalized Nash equilibrium problems\\ in aggregative form}
\label{sec:constrained}

\subsection{\algoc/}

In this section, we introduce the \algoc/ algorithm, i.e., a distributed iterative methodology to find a GNE in aggregative games with affine coupling constraints as formalized in~\eqref{eq:problem_constrained}.

In addition to the assumptions made in Section~\ref{sec:setup}, we need some further conditions for our mathematical developments.
\begin{assumption}[Feasibility]
	\label{ass:feasible_set}
	The set $\cC$ is nonempty.%
	\QEDB
\end{assumption}
Note that the condition $\cC \ne \emptyset$ is weaker than Slater's constraint qualification required by many results in the literature.
However, to establish linear convergence of our distributed algorithm, we will enforce an additional assumption on the matrix $A$ (see Assumption~\ref{ass:feasible_set_primal_duality}).
Consider the following variational inequality, defined by the mapping $F$ in \eqref{eq:F_definition} over the domain $\cC$: %
\begin{equation}\label{eq:VI_con}
    F(\xstar)\T (\x - \xstar) \ge 0, ~\text{ for all } \x \in \cC.
\end{equation}
It is known that every point $\xstar \in \cC$ for which~\eqref{eq:VI_con} holds is a GNE of the game~\eqref{eq:problem_constrained} and, specifically, a \emph{variational} GNE (v-GNE) (cf.~\cite[Th.~3.9]{FacchineiKanzowGNE2010}).
The converse, however, does not hold in general.
However, since $F$ is strongly monotone (cf. Standing Assumption~\ref{ass:objective_function}) and $\cC$ is nonempty (cf. Assumption~\ref{ass:feasible_set}), closed and convex (since the constraint are in the form $Ax \leq b$), Prop.~1.4.2 and Th.~2.3.3 in~\cite{facchinei2003finite} guarantee that a unique v-GNE exists, and this satisfies \eqref{eq:VI_con}.

In the following, we devise an iterative algorithm that asymptotically returns the (unique) v-GNE of~\eqref{eq:problem_constrained}. 
Inspired by~\cite{qu2018exponential}, where an augmented primal-dual scheme was used for continuous-time, centralized optimization, we require the following additional condition on the matrix $A$ which characterizes the coupling constraints (cf.~\eqref{eq:coupl_constr}):
    \begin{assumption}[Full-row rank]
        \label{ass:feasible_set_primal_duality}
        There exist $\kappa_1, \kappa_2 > 0$ such that $\kappa_1 I_m \preccurlyeq AA\T \preccurlyeq \kappa_2 I_m$.
        \QEDB
    \end{assumption}
    We note that Assumption~\ref{ass:feasible_set_primal_duality} imposes, as a necessary condition, the fact that $m \leq n$, i.e., that the number of constraints is at most equal to the total number of components of the global strategy vector.

    Following~\cite{qu2018exponential}, for all $i \in \mc I$ we consider the augmented Lagrangian function $L_i: \R^{n} \times\R^m \to \R$ defined as 
    \begin{equation}\label{eq:lagrangian}
        L_i(\x,\lambda) \coloneqq \Ji(\x_i,\sigma(\x)) + \underbrace{\sum_{\ell=1}^m H_\ell([A\x - b]_\ell,[\lambda]_\ell)}_{\eqqcolon H(A\x - b,\lambda)},
    \end{equation}
	where
    \begin{align*}
        &H_\ell([A\x - b]_\ell,[\lambda]_\ell) \coloneqq
        \begin{cases}
        \hspace{-.01cm}[A\x - b]_\ell[\lambda]_\ell + \frac{\rho}{2}([A\x - b]_\ell)^2 \hspace{.23cm} \text{if} \hspace{.09cm} \rho[A\x - b]_\ell + [\lambda]_\ell \ge 0
        \\
        \hspace{-.01cm}-\frac{1}{2\rho}[\lambda]_\ell^2 \hspace{3.44cm} \text{if} \hspace{.09cm} \rho[A\x - b]_\ell + [\lambda]_\ell < 0,
        \end{cases}
    \end{align*}
    with $\lambda \in \R^m$ being the multiplier associated to the coupling constraints, and $\rho > 0$ a constant. We therefore address the v-GNE seeking problem by obtaining a saddle point of~\eqref{eq:lagrangian} through the discrete-time dynamics:
    \begin{subequations}\label{eq:local_desired_primal_dual_implicit}
    \begin{align}
        \xitp  &=  \xit  -  \delta\left(\nabla_{\x_i}\Ji(\xit, \sigma(\xt))  +  \nabla_{\x_i} H(A\xt -  b,\lt)\right)
        \\
        \ltp  &=  \lt +\delta\nabla_{\lambda} H(A\xt - b,\lt),
    \end{align}
    \end{subequations}
    where $\xit$ and $\delta$ have the same meaning as in~\eqref{eq:desired_update}, $\lt \in \R^m$ is the multiplier at $\iter \ge 0$, and the explicit form of the gradients $\nabla_{\x_i} H(A\xt - b,\lt)$ and $\nabla_{\lambda} H(A\xt - b,\lt)$ reads as
    \begin{subequations}\label{eq:expl_nablaH}
    \begin{align}
        \nabla_{\x_i} H(A\xt - b,\lt) &= \sum_{\ell=1}^m\nabla_{\x_i}H_\ell([A\xt - b]_\ell,[\lt]_\ell) 
        \notag\\
        &= \sum_{\ell=1}^m \max\left\{\rho[A\xt - b]_\ell + [\lt]_\ell,0\right\}[A_i]_\ell\T %
        \\
        \nabla_{\lambda} H(A\xt - b,\lt)
        &= \sum_{\ell=1}^m \nabla_{\lambda}H_\ell([A\xt - b]_\ell,[\lt]_\ell) 
        \notag\\
        &= \sum_{\ell=1}^m \frac{1}{\rho} e_\ell (\max\left\{\rho[A\xt - b]_\ell + [\lt]_\ell,0\right\} - [\lt]_\ell),
    \end{align}
\end{subequations}
	where $e_\ell \in \R^m$ is the $\ell$-th vector of the canonical basis of $\R^m$, $\ell \in \{1, \ldots, m\}$. The stacked-column form of~\eqref{eq:local_desired_primal_dual_implicit} is 
    \begin{subequations}\label{eq:desired_primal_dual_implicit}
        \begin{align}
            &\xtp = \xt- \delta\left( F(\x) + \nabla_{\x} H(A\xt - b,\lt)\right)
            \\
            &\ltp = \lt +\delta \nabla_{\lambda} H(A\xt - b,\lt),
        \end{align}
        \end{subequations}
        where $\nabla_{\x} H(A\xt - b,\lt) \coloneqq \col(\nabla_{\x_1} H(A\xt - b,\lt),\dots,\nabla_{\x_N} H(A\xt - b,\lt))$.
    By computing the KKT conditions of the VI~\eqref{eq:VI_con} and using~\cite[Prop.~1]{qu2018exponential}, we obtain that the v-GNE $\xstar$ and the corresponding (unique) optimal multiplier $\lstar \in \R^m$ are such that 
    \begin{subequations}\label{eq:equilibrium_primal_dual}
        \begin{align}
            0 &= F(\xstar) + \nabla_{\x} H(A\xstar - b,\lstar)
            \\
            0 &= \nabla_{\lambda} H(A\xstar - b,\lstar).
        \end{align}
    \end{subequations}  
    The above result ensures that $\col(\xstar,\lstar)$ represents an equilibrium point of~\eqref{eq:desired_primal_dual_implicit} for any $\delta > 0$.
    
    However, since agent $i$ does not have access neither to $\sigma(\xt)$ nor to $A\xt - b$, the scheme in~\eqref{eq:local_desired_primal_dual_implicit} cannot be directly implemented.
    Moreover, dynamics~\eqref{eq:local_desired_primal_dual_implicit} requires a central unit that can compute the global quantity $A\xt - b$ and communicate the multiplier $\lt$ to all the agents. 
    For this reason, in Algorithm~\ref{algo:constrained} we introduce for all $i \in \mc I$ (i) two additional variables $\z_i \in \R^d$ and $\y_i \in \R^m$ to compensate the local unavailability of $\sigma(\xt)$ and $A\xt - b$, respectively, (ii) a local copy $\lambda_i \in \R^m$ of the multiplier $\lambda$, and (iii) an additional averaging step to enforce consensus among the multipliers $\lambda_i$ (cf.~\eqref{eq:l_local_constrained_primal_dual}-\eqref{eq:y_local_constrained_primal_dual}).
    As already done in~\eqref{eq:z_local_update}, we choose causal perturbed consensus dynamics to update $\z_i$ and $\y_i$. For all $i \in \mc I$, we then introduce operators $\Gxi: \R^{m} \times \R^m \to \R^{n_i}$ and $\Gli: \R^{m} \times \R^m \to \R^m$ as
\begin{align}\label{eq:eqG}
    \begin{split}
	\Gxi(s_1,s_2) &\coloneqq \sum_{\ell=1}^m \textrm{max}\{\rho[s_1]_\ell + [s_2]_\ell,0\}[A_i]\T_\ell %
	\\
	\Gli(s_1,s_2) &\coloneqq \frac{1}{\rho}\sum_{\ell=1}^m \left(  \textrm{max}\{\rho[s_1]_\ell + [s_2]_\ell,0\} - [s_2]_\ell \right) e_\ell.   
   \end{split}
\end{align}
In Algorithm \ref{algo:constrained}, these operators encode the component of the gradients in \eqref{eq:expl_nablaH} available to agent $i$ at iteration $t$, plus the auxiliary variable $\yit$ that is used to track $A\xt - b$ (see \eqref{eq:x_local_constrained_primal_dual} and \eqref{eq:l_local_constrained_primal_dual} in Algorithm \ref{algo:constrained}).
The steps of the proposed method are hence summarized in Algorithm~\ref{algo:constrained} from the perspective of agent $i$, which is then referred as \algoc/. 
Note that all the quantities involved in the agent's calculations are purely local, thus making Algorithm~\ref{algo:constrained} fully distributed.
    \begin{algorithm}[t!]
        \begin{algorithmic}
            \State \textbf{Initialization}: $\x_i^0 \in \X_i, \lit \in \R^m_{+}, \z_i^0 = 0, \y_i^0 = 0$.
            \For{$t=0, 1, \dots$}
                \begin{subequations}\label{eq:local_primal_dual}
                    \begin{align}
                        \xitp &= \xit - \delta(\tFi(\xit, \phii(\xit) + \zit) + (N(A_i\xit - b_i) + \yit,\lit))
                         \label{eq:x_local_constrained_primal_dual}
                        \\
                        \litp &= \sum_{j \in \cN_i} w_{ij}\ljt +\delta \Gli(N(A_i\xit - b_i) + \yit,\lit)
                        \label{eq:l_local_constrained_primal_dual}
                        \\
                        \zitp &= \sum_{j\in\cN_i}w_{ij}\zjt + \sum_{j\in\cN_i}w_{ij}\phij(\xjt) - \phii(\xit)\label{eq:z_local_constrained_primal_dual}
                        \\
                        \yitp &= \sum_{j\in\cN_i}w_{ij}\yjt + \sum_{j\in\cN_i}w_{ij}N(A_j\xjt - b_j) 
                        - N(A_i\xit - b_i),\label{eq:y_local_constrained_primal_dual}
                    \end{align}
                \end{subequations}
            \EndFor
        \end{algorithmic}
        \caption{\algoc/ (Agent $i$)}
        \label{algo:constrained}
    \end{algorithm}
    Differently from customary primal-dual schemes, \eqref{eq:l_local_constrained_primal_dual} does not need the projection over the positive orthant $\R^m_{+}$ due to the chosen augmented Lagrangian functions $L_i$~\eqref{eq:lagrangian}.
    We only need to initialize $\lambda_i^0 \ge 0$ for all $i \in \mc I$, and choose $\delta$ and $\rho$ appropriately so that we avoid situations where $\lit \geq 0$ implies $\litp < 0$. 
    To see this notice first that if $\lit = 0$, then it is easy to check $\Gli(N(A_i\xit - b_i) + \yit,\lit) \ge 0$ and, thus, $\litp \ge 0$.
    The critical scenario for agent $i$ occurs when all the multipliers of its neighbors are zero, namely $\ljt = 0$ for any $j \in \cN_i$, and when $\max\{\rho[N(A_i\xit - b_i) + \yit]_{\ell} + [\lit]_\ell, 0\} = 0$ for at least one $\ell \in \until{m}$.
    Indeed, specializing \eqref{eq:l_local_constrained_primal_dual} for this case leads to the following update of that $\ell$-th component of $\lit$
    \begin{align}
        [\litp]_{\ell} = \left(w_{ii} - \frac{\delta}{\rho}\right)[\lit]_{\ell}.\label{eq:critical_update}
    \end{align}
    From~\eqref{eq:critical_update}, we conclude that $[\litp]_{\ell} $ remains non-negative if $[\lit]_{\ell}$ is non-negative, thus alleviating the need for a projection, as long as $\delta$ and $\rho$ satisfy $w_{ii} > \delta/\rho$.
    This feature plays a key role in proving exponential stability properties for the continuous-time, centralized primal-dual scheme proposed in~\cite{qu2018exponential}.
    As in the case without coupling constraints, the purpose of the initialization step will become clear in the next subsection. The steps of Algorithm \ref{algo:constrained} in \eqref{eq:local_primal_dual} can be compactly written as:
    \begin{subequations}\label{eq:algo_primal_dual}
        \begin{align}
            \xtp &= \xt + \delta f_{x}(\xt,\lt,\zt,\yt)
            \\
            \ltp &= \Wm\lt + \delta\Gl(N(\bar{A}\xt - \bar{b}) + \yt, \lt)
            \\
            \ztp &= \Wd\zt + (\Wd - I)\phi(\xt)
            \\
            \ytp &= \Wm\yt + (\Wm - I)N(\bar{A}\xt - \bar{b}).
        \end{align}
    \end{subequations}
    where $\map{f_{x}}{\R^{n} \times \R^{Nm} \times \R^{Nd} \times \R^{Nm}}{\R^{n}}$ is defined as
    \begin{align*}
        f_{x}(\x,\lambda,\z,\y)\coloneqq -\tF(\x, \phi(\x) + \z) - \Gx(N(\bar{A}\x - \bar{b}) + \y, \lambda),
    \end{align*}
and, similarly to~\eqref{eq:global_update}, $\lambda : \col(\lambda_1,\dots,\lambda_N)$, $\Wd \coloneqq \cW \otimes I_d$, $\Wm \coloneqq \cW \otimes I_m$, $\Gx(N(\bar{A}\xt - \bar{b}) + \yt, \lt) \coloneqq \col(G_{\x,1}(N(A_1\x_1^\iter - b_1) + \y_1^\iter,\lambda_1^\iter),\dots,G_{\x,N}(N(A_N\x_N^\iter - b_N) + \y_N^\iter,\lambda_N^\iter))$, and $\Gl(N(\bar{A}\xt - \bar{b}) + \yt, \lt) \coloneqq \col(G_{\lambda,1}(N(A_1\x_1^\iter - b_1) + \y_1^\iter,\lambda_1^\iter),\dots,G_{\lambda,N}(N(A_N\x_N^\iter - b_N) + \y_N^\iter,\lambda_N^\iter))$.
    Next, we establish the convergence properties of \algoc/ in computing the v-GNE of~\eqref{eq:problem_constrained}.
    \begin{theorem}\label{th:convergence_primal_dual}
	Consider~\ref{eq:algo_primal_dual} and Assumptions~\ref{ass:feasible_set}, \ref{ass:feasible_set_primal_duality}. Let $(\x^0,\lambda^0,\z^0,\y^0) \in \X \times \R^{Nm}_+ \times \R^{Nd} \times \R^{Nm}$ satisfy $\oned\T \z^0 = 0$ and $\onem\T \y^0 = 0$. Then, there exist $\bar{\delta}, a_1, a_2 > 0$ such that, for any $\delta \in (0,\bar{\delta})$, with $w_{ii} > \frac{\delta}{\rho}$ for all $i \in \until{N}$, it holds
	\begin{align*}
		\norm{\xt - \xstar} \leq a_1 e^{-a_2\iter}.\eqoprocend
	\end{align*}
	\end{theorem}
    Note that the additional condition $w_{ii} > \delta/\rho$ needs to be satisfied by $\delta$, given $\rho$, to ensure the dual variables remain non-negative, as discussed below \eqref{eq:critical_update}.
	As in the case of NE seeking without coupling constraints, the proof of Theorem~\ref{th:convergence_primal_dual} relies on a \emph{singular perturbations} analysis of system~\eqref{eq:algo_primal_dual}. 
    We provide this in the next subsection.
 \subsection{Proof of Theorem \ref{th:convergence_primal_dual}}
	\label{sec:sp_primal_dual}

As with the proof of Theorem \ref{th:convergence}, we show that the setting of Theorem \ref{th:convergence_primal_dual} fits the framework of Theorem \ref{th:theorem_generic}, and organize its proof in five steps.

\emph{1. Bringing \eqref{eq:algo_primal_dual} in the form of \eqref{eq:interconnected_system_generic}:}
    we introduce the change of coordinates
    \begin{align}\label{eq:change_mean_primal_dual}
        \begin{bmatrix}
            \bzt\\
            \pzt\\
        \end{bmatrix} &= \begin{bmatrix}
            \frac{\oned\T}{N}\\
            \Rd\T 
        \end{bmatrix}\zt, \quad 
        \begin{bmatrix}
            \byt\\
            \pyt\\
        \end{bmatrix} = \begin{bmatrix}
            \frac{\onem\T}{N}\\
            \Rm\T 
        \end{bmatrix}\yt 
        \notag\\ 
        \begin{bmatrix}
            \blt\\
            \plt
        \end{bmatrix} &= \begin{bmatrix}
            \frac{\onem\T}{N}\\
            \Rm\T 
        \end{bmatrix}\lt,
    \end{align}
    where $\Rd \in \R^{Nd \times (N-1)d}$, $\Rd\T\Rd = I$, $\Rm \in \R^{Nm \times (N-1)m}$, $\Rd\T\Rd = I$, $\norm{\Rd} = 1$, $\norm{\Rm} = 1$, and
        \begin{align}
    \Rd\Rd\T = I - \frac{\oned\oned\T}{N},~ \Rm\Rm\T = I - \frac{\onem\onem\T}{N} \label{eq:RdRm}.
    \end{align}
    
As in the proof of Theorem~\ref{th:convergence_primal_dual}, we use the initialization $\oned\T\z^0 = 0$ and $\onem\T\y^0 = 0$ to ensure that $\bzt = 0$ and $\byt = 0$ for all $\iter \ge 0$. 
    In view of~\eqref{eq:change_mean_primal_dual}, we can therefore rewrite~\eqref{eq:algo_primal_dual} by ignoring the dynamics of $\bzt$ and $\byt$, thus obtaining the system
    \begin{subequations}\label{eq:algo_primal_dual_compact}
        \begin{align}
            \chitp &= \chit + \delta f(\chit,\wt)\label{eq:chi_update_primal_dual}
            \\
            \wtp &= S\wt + K(\delta)u(\chit).\label{eq:w_update}
        \end{align}
    \end{subequations}
   in which
    \begin{subequations}
    \begin{align}
        \chit &\coloneqq \begin{bmatrix}
            \xt\\
            \blt
        \end{bmatrix}, \quad \wt \coloneqq \begin{bmatrix}
            \plt\\
            \pzt\\
            \pyt
        \end{bmatrix}
        \\
        f(\chit,\wt)
         &\coloneqq \begin{bmatrix}
            f_{x}(\xt,\onem\blt + \Rm\plt,\Rd\pzt,\Rm\pyt)
            \\
            \frac{\onem\T}{N}\Gl(N(\bar{A}\xt - \bar{b}) + \Rm\pyt, \onem\blt + \Rm\plt)
        \end{bmatrix}\label{eq:f_primal_dual} 
        \\
        S &\coloneqq \begin{bmatrix}
            \Rm\T \Wm\Rm& 0& 0\\
            0& \Rd\T \Wd\Rd& 0\\
            0& 0& \Rm\T \Wm\Rm
        \end{bmatrix}
        \\
        K(\delta) &\coloneqq \begin{bmatrix}
            \delta \Rm\T& 0& 0
            \\
            0& \Rd\T(\Wd - I)& 0\\
            0& 0& \Rm\T(\Wm - I)
        \end{bmatrix}    
        \\
        u(\chit) &\coloneqq \begin{bmatrix}
            \Gl(N(\bar{A}\xt - \bar{b}) + \Rm\pyt, \onem\blt + \Rm\plt)
            \\
            \phi(\xt)
            \\
            N(\bar{A}\xt - \bar{b})
        \end{bmatrix}.
            \end{align}
    \end{subequations}
    We view~\eqref{eq:algo_primal_dual_compact} as a singularly perturbed system, namely the interconnection between the slow dynamics~\eqref{eq:chi_update_primal_dual} and the fast one~\eqref{eq:w_update}. 
        Indeed, system~\eqref{eq:algo_primal_dual_compact} can be obtained from~\eqref{eq:interconnected_system_generic} by considering $\chit$ as the state of~\eqref{eq:slow_system_generic} and setting
        \begin{align}
            g(\chit,\wt,\delta) := S\wt + K(\delta)u(\chit).\label{eq:g_primal_dual} 
        \end{align}
    
    \emph{2. Equilibrium function $h$:}
    under the double stochasticity condition of $\cW$, due to Standing Assumption~\ref{ass:network}, and using \eqref{eq:RdRm},
    for any $\chit = \chi$,
    \begin{align}\label{eq:h_primal_dual_definition}
        h(\chi) \coloneqq \begin{bmatrix}
            0
            \\
            -\Rd\T\phi\left(\begin{bmatrix}
                I_n& 0
            \end{bmatrix}\chi\right)
            \\
            -\Rm\T N\left(\bar{A}\begin{bmatrix}
                I_n& 0
            \end{bmatrix}\chi - \bar{b}\right)
        \end{bmatrix}
    \end{align}
   constitutes an equilibrium of~\eqref{eq:w_update} (parametrized by $\chi$).
    
    \emph{3. Boundary layer system and satisfaction of \eqref{eq:U_generic}:}
the so-called boundary layer system associated to \eqref{eq:algo_primal_dual_compact} can be constructed by fixing $\chit=\chi=\col(\x,\bl)$ for some arbitrary $(\x,\bl) \in \R^n \times \R^m$, and rewriting it according to the error coordinates $\tw \coloneqq \col(\tpl,\tilde{z}_{\perp},\tilde{y}_{\perp}) \coloneqq w - h(\chi)$. Using \eqref{eq:RdRm}, we then obtain that
    \begin{equation}\label{eq:bl_primal_dual}
        \twtp = S\twt + \delta \tilde{u}(\chi,\twt),
    \end{equation}
    where
    \begin{align*}
        &\tilde{u}(\chi,\twt) 
        \coloneqq\begin{bmatrix}\Rm\T\Gl\left(\onem(A\x -b) + \Rm\tpyt, \onem\bl + \Rm\tplt\right)
            \\
            0
            \\
            0
        \end{bmatrix}.
    \end{align*}
    The next lemma provides a Lyapunov function for~\eqref{eq:bl_primal_dual}.
    \begin{lemma}\label{lemma:bl_primal_dual}
        Consider system~\eqref{eq:bl_primal_dual}. Then, there exists a continuous function $U: \R^{(N-1)(2m + d)} \to \R$ and $\bar{\delta}_1 > 0$ such that, for any $\delta \in (0,\bar{\delta}_1)$, $U$ satisfies~\eqref{eq:U_generic} with $\tw$ in place of $\psi$.\QEDB
    \end{lemma}

 \emph{4. Reduced system and satisfaction of \eqref{eq:W_generic}:}
 the so-called reduced system can be obtained by considering the fast dynamics in~\eqref{eq:chi_update_primal_dual} at steady state, i.e., $\wt = h(\chit)$ for any $\iter \ge 0$. We thus have
    \begin{equation}\label{eq:rs_primal_dual}
        \chitp = \chit + \delta f(\chit,h(\chit)).
    \end{equation}
    Let us expand~\eqref{eq:rs_primal_dual}. Using \eqref{eq:RdRm}, we obtain
    \begin{subequations}\label{eq:rs_primal_dual_explicit}
        \begin{align}
            \xtp &= \xt - \delta\left(\tF\left(\xt, \oned\sigma(\xt)\right)+\Gx\left(\onem(A\xt - b), \onem\blt\right)\right)\label{eq:x_update_rs}
            \\
            \bltp &= \blt + \delta\frac{\onem\T}{N}\Gl\left(\onem(A\xt - b), \onem\blt\right).\label{eq:barlambda_update_rs}
        \end{align}
    \end{subequations}
    Notice that
    \begin{align*}
        \tF\left(\x, \oned\sigma(\x)\right) &= F(\x)
        \\
        \Gx\left(\onem(A\xt - b), \onem\blt\right) &= \nabla_{\x} H(A\xt - b, \blt),
    \end{align*}
    and also
    \begin{align*}
        \frac{\onem\T}{N}\Gl\left(\onem(A\xt - b), \onem\blt\right) = \nabla_{\lambda} H(A\xt - b, \blt).
    \end{align*}
    Therefore, \eqref{eq:rs_primal_dual} is identical to the original update~\eqref{eq:desired_primal_dual_implicit}. Given the unique v-GNE $\xstar$ of~\eqref{eq:problem_constrained} (see Assumptions~\ref{ass:feasible_set},~\ref{ass:feasible_set_primal_duality}) and the associated multiplier $\lstar \in \R^m$, the next lemma provides a Lyapunov function for \eqref{eq:rs_primal_dual}, hence for~\eqref{eq:desired_primal_dual_implicit}. 
    \begin{lemma}\label{lemma:rs_primal_dual}
        Consider system~\eqref{eq:rs_primal_dual} and Assumptions~\ref{ass:feasible_set},~\ref{ass:feasible_set_primal_duality}. Then, there exist a continuous function $W: \R^{n + m} \to \R$, $\bar{\delta} > 0$ such that, for any $\delta \in (0,\bar{\delta})$, $W$ satisfies~\eqref{eq:W_generic} with $\chi$ in place of $x$.\QEDB
    \end{lemma}

\emph{5. Lipschitz continuity of $f$, $g$ and $h$:}
as we will be invoking Theorem~\ref{th:theorem_generic}, we need to ensure that the required Lipschitz properties are satisfied. In particular, we need to show that $f$, $g$ in \eqref{eq:f_primal_dual} and \eqref{eq:g_primal_dual}, respectively, and $h$ in \eqref{eq:h_primal_dual_definition} are Lipschitz with respect to their arguments. This is guaranteed by the Lipschitz continuity of the aggregation rules and the gradients of the cost functions (cf. Standing Assumption~\ref{ass:objective_function}), %
and the Lipschitz continuity of $G_x$ and $G_{\lambda}$ (that appear in $f$ and $g$), which is ensured as shown in \eqref{eq:max_r1_r2_bl} within the proof of Lemma\ref{lemma:bl_primal_dual}.

    By combining Lemmas~\ref{lemma:bl_primal_dual} and \ref{lemma:rs_primal_dual} with the Lipschitz continuity properties expressed above, we can apply Theorem~\ref{th:theorem_generic}.
     Then, there exists $\bar{\delta} \in (0, \min(\bar{\delta}_1,\bar{\delta}_2))$ so that, for any $\delta \in (0,\bar{\delta})$, $\col(\xstar,\lstar,h(\xstar,\lstar))$ is an exponentially stable equilibrium point for~\eqref{eq:algo_primal_dual_compact}.

\section{Numerical Examples}
\label{sec:numerical_simulation}

We demonstrate the efficacy of \algo/ and \algoc/ and compare them with the most closely related distributed equilibrium-seeking algorithms from the literature.
First, we consider the case with local constraints only, and then we focus also on problems with coupling constraints.
In both cases, we performed Monte Carlo simulations consisting of $25$ trials.
In each trial, we randomly generate the problem parameters, the graph of the network, and the initial conditions of the algorithms' variables.

\subsection{Example without coupling constraints}

In this subsection, we consider an instance of problem~\eqref{eq:problem} and perform numerical simulations in which we compare \algo/ with Algorithm~2 proposed in~\cite{parise2020distributed} and Algorithm~4 proposed in~\cite{bianchi2022fast}.
We consider the multi-agent demand response problem considered in~\cite{parise2020distributed}. Consider $N$ loads whose electricity consumption $\x_i \coloneqq \col(\x_{i,1},\dots,\x_{i,T}) \in \R^T$ with $T \in \mathbb{N}$ has to be chosen to solve
\begin{equation*}
	\forall i \in \mc I : \min_{\x_i \in \X_i} \, \rho_i\norm{\x_i - \hat{u}_i}^2 + (\lambda\sigma(\x) + p_0)\T \x_i,
\end{equation*}
where $\hat{u}_i \in \R^T$ denotes some nominal energy profile, $\rho_i > 0$ is a constant weighting parameter, and the term $\lambda\sigma(\x) + p_0$  with $\lambda \in \R$, $p_0 \in \R^T$ models the unit price which is taken to be an affine increasing function of the aggregate (average) energy demand $\sigma(\x)= (1/N) \sum_{i \in \mc I} \x_i$. 
As for the local feasible set $\X_i \subseteq \R^{T}$, for all $i \in \mc I$, we pick 
\begin{align*}
    \X_i \coloneqq \bigg\{\x_i \in \R^T  \mid\,  &s_{i,\tau+1}(\x_i) \in \mc S_i\; \text{and}\; \x_{i,\tau} \in \mc U_i \hspace{0.1cm} \forall \tau \in \until{T},
     \sum_{\tau=1}^T \x_{i,\tau} =  \sum_{\tau=1}^T \hat{u}_{i,\tau}\bigg\},
\end{align*}
where $\mc U_i \subseteq \R$, $\mc S_i \subseteq \R$, and $s_{i,\tau}(\x_i)$ is the state of the $i$-th load at time $\tau$; this, given the parameters $a_i, b_i \in \R$, is computed according to the linear dynamics
\begin{align*}
    s_{i,\tau} = a_i^{\tau-1}s_{i,1} + \sum_{k=1}^{\tau - 1}a^{k-1}b_i\x_{i,\tau - k},
\end{align*}
where $s_{i,1} \in \mc S_i$ is the initial condition of the state of the $i$-th load.
To instantiate the problem, we set $T = 24$ and randomly generate values for $\hat{u}_i$, $\rho_i$, $\lambda$, $p_0$, $a_i$, $b_i$, $s_{i,1}$ and initial strategies $\x_{i,1}$ from uniform distributions. As for the sets $\mc U_i$ and $\mc S_i$, we pick the intervals $[0,1]$ and $[0,10]$, respectively. We consider a network with $N = 10$ players communicating according to an undirected, connected \er/ graph with parameter $0.3$. 

This setting satisfies our standing assumptions. 
We compare our scheme, namely, \algo/ with Algorithm~2 in~\cite{parise2020distributed} and Algorithm~4 in~\cite{bianchi2022fast}.
We empirically tune the former with $v_1 = v_2 = 20$ communication rounds per iterate and update the auxiliary variable $\zt$ according to $\ztp = (1-\lambda)\zt + \lambda\mathcal{A}_{v_1,v_2}$ with $\lambda = 0.01$ (the quantity $\mathcal{A}_{v_1,v_2}$ is a proxy for the unavailable aggregative variable $\sigma(\x)$, see~\cite{parise2020distributed} for more details). 
We empirically tuned the method by~\cite{bianchi2022fast} choosing $\alpha = 0.1$, $\beta = 1$, and $\tau_i = 0.1$ for all $i \in \mc I$.
As for the parameters of our scheme, we set $\delta = 0.5$ and $\gamma = 0.001$.
Fig.~\ref{fig:error_dr} shows the evolution of the normalized distance $\norm{\xt - \xstar}/\norm{\xstar}$ from the NE $\xstar$ as the communication rounds (corresponding to iterations) progress. Our algorithm exhibits faster convergence and achieves higher accuracy in the calculation of the equilibrium $\xstar$ with respect to the method in~\cite{parise2020distributed}, while it turns out to be slower than the algorithm in~\cite{bianchi2022fast}.
This was anticipated as (i) the method in~\cite{parise2020distributed} is not guaranteed to converge to the exact NE (see Table \ref{table:unconstrained}) and (ii) the method in~\cite{bianchi2022fast} is based on proximal-based updates which are known to exhibit faster behavior compared to gradient-based updates, but are computationally more intensive due to the proximal operator involved.
In Table~\ref{table:execution_times}, we provide a numerical comparison of the considered methods in terms of the mean and standard deviation based on Monte  Carlo simulations of the time needed to perform a single iterate.
As expected, \algo/ turns out to be much lighter than the other algorithms from a computational point of view. 
The simulations have been executed on Matlab, using {\scshape fmincon()} to solve the optimization steps involved in the methods by~\cite{parise2020distributed} and~\cite{bianchi2022fast}.
\begin{figure}
    \begin{minipage}{\columnwidth}
        \centering
        \includegraphics[scale=.9]{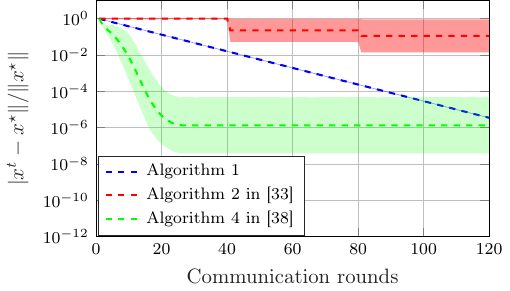}
	\caption{Mean and $1$-standard deviation band (based on Monte Carlo simulations) of the normalized distance of the iterates from the NE achieved by \algo/ (Algorithm~\ref{algo:unconstrained}), the algorithm by~\cite{parise2020distributed}, and the algorithm by~\cite{bianchi2022fast} on a case study introduced in \cite{parise2020distributed}.}
	\label{fig:error_dr}
    \end{minipage}
    \begin{minipage}{\columnwidth}
        \captionsetup{type=table} 
        \centering
        \vspace{1em}
            \begin{tabular}{c|c|c|c}  
                Time of iterate &Alg.~1 &Alg.~2 in~\cite{parise2020distributed} &Alg.~4 in~\cite{bianchi2022fast} \\ 
                \hline
                Mean (s) &$8 \times 10^{-4}$ &$6.38 \times 10^{-2}$ &$7.67 \times 10^{-2}$
                \\
                \hline
                Std. Dev. (s) &$2 \times 10^{-4}$ &$9.6 \times 10^{-3}$ &$1.18 \times 10^{-2}$\\
                \hline
            \end{tabular}
        	\caption{Execution time of a single iterate with \algo/ (Algorithm~\ref{algo:unconstrained}) and the considered algorithms in~\cite{parise2020distributed} and~\cite{bianchi2022fast}. Mean and $1$-standard deviation are based on Monte Carlo simulations of the case study in \cite{parise2020distributed}.}\label{table:execution_times}
    \end{minipage}
\end{figure}

\subsection{Example with coupling constraints}

Here, we compare our \algoc/ algorithm with the distributed methods proposed in~\cite{belgioioso2020distributed},~\cite{gadjov2020single} and~\cite{bianchi2022fast}.
For a fair comparison, we test the scheme by~\cite{belgioioso2020distributed} with a constant step-size even if convergence was theoretically proven only with a diminishing one (see Table \ref{table:constrained}); note that slower convergence is expected by using a diminishing step-size.
We focus on a case study inspired by~\cite{li2021distributed} -- where it was addressed within a cooperative scenario -- and adapt it as an instance of \eqref{eq:problem_constrained}. 
In particular, we consider the cost function
\begin{align*}
    J_i(\x_i,\sigma(\x)) = \frac{1}{2}\norm{\x_i - p_i}^2 + \frac{w}{2}\norm{\x_i - \sigma(\x)},
\end{align*} 
where $w > 0$ and $p_i \in \R^{n_i}$ for all $i \in \mc I$, while $\sigma(x) = \frac{1}{N}\sum_{i\in \mc I}\x_i$. 
We consider a communication graph with ring topology.
As for the coupling constraints, in each trial of the Monte Carlo simulations, we randomly generate each $A_i$ and $b_i$ by imposing the full row rank property for the former (cf. Assumption~\ref{ass:feasible_set_primal_duality}) and extracting the latter from the interval $[0,100]$ with a uniform probability; we set $N = 20$, $n_i = 2$.
Moreover, in each trial, we uniformly randomly extract each $p_i$ and $w$ from $[0,100]^{2}$ and $[0,1]$, respectively.
We empirically tune the algorithm in~\cite{belgioioso2020distributed} with $\alpha_i = \beta_i = 0.5$ for all $i \in \mc I$, and $\gamma^\iter = 0.1$ for all $\iter \ge 0$.
As for the parameters of the method in~\cite{gadjov2020single}, we empirically choose $c = 1$, $k = 0.1$, $\tau = 0.2$, $\alpha = 0.2$, and $v = 0.1$.
The algorithm in~\cite{bianchi2022fast} has been empirically tuned setting $\alpha = 0.3$, $\beta = 0.1$, $\tau_i = 0.3$ and $\delta_i = 0.3$ for all $i \in \mc I$.
Finally, as for the parameters of our algorithm, we empirically tune them as $\delta = 0.05$ and $\rho = 0.1$.
In Fig.~\ref{fig:montecarlo_constraints}, we compare the performance of the algorithms in~\cite{belgioioso2020distributed},~\cite{gadjov2020single}, and~\cite{bianchi2022fast} with Algorithm~\ref{algo:constrained} in terms of the normalized distance $\norm{\xt - \xstar}/\norm{\xstar}$ from the GNE $\xstar$. 
In this case, the proposed scheme outperforms the others in terms of accuracy and convergence speed.
\begin{figure}
	\centering
	\includegraphics[scale=.9]{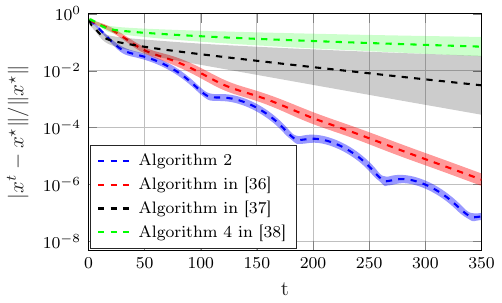}
	\caption{Mean and $1$-standard deviation band (based on Monte Carlo simulations) of the normalized distance of the iterates from the GNE achieved by \algoc/ (Algorithm~\ref{algo:constrained}), and the algorithms by~\cite{belgioioso2020distributed},~\cite{gadjov2020single}, and~\cite{bianchi2022fast}.}.
	\label{fig:montecarlo_constraints}
\end{figure}

\section{Conclusion and Outlook}
\label{sec:conlcusions}

We propose two novel fully-distributed algorithms for (generalized) equilibrium seeking in aggregative games over networks. 
The first algorithm is designed to address the case where only local constraints are present.
The second method does not involve local constraints, however, it allows handling coupling constraints, thus encompassing generalized Nash equilibrium problems.
Both schemes are studied by means of singular perturbations analysis in which slow and fast dynamics are identified and separately investigated to demonstrate the linear convergence of the whole interconnection to the (generalized) Nash equilibrium. 
Current work concentrates on extending our analysis line to allow for local constraint sets, either in a hard manner or by means of dualizing these and satisfying them asymptotically.
An additional aspect worth of investigation is the possibility of time-varying communication patterns among the agents.
Finally, we perform detailed numerical simulations showing the effectiveness of the proposed methods and that they outperform state-of-the-art distributed methods. 

\appendix

\setcounter{section}{0}
\renewcommand{\thesection}{\Alph{subsection}}

\subsection{$Q$-Linear rate}
\label{sec:linear}

Here, we report the definition of $Q$-linear convergence~\cite[App.~A.2]{nocedal1999numerical}.
For the sake of readability, in the rest of the document, we omit the prefix $Q$.
\begin{definition}
    Let $\{\xt\}$ be a sequence in $\R^n$ that converges to $\xstar \in \R^{n}$.
    We say that the convergence is $Q$-linear if there is a constant $r \in (0, 1)$ such that 
    \begin{align*}
        \frac{\norm{\xtp - \xstar}}{\norm{\xt - \xstar}} \leq r,
    \end{align*}
    for all $\iter$ sufficiently large.\QEDB
\end{definition}

\subsection{Proof of Theorem \ref{th:theorem_generic}}\label{sec:auxiliary}

    Let $\twt \coloneqq \wt - h(\xt)$ and accordingly rewrite~\eqref{eq:interconnected_system_generic} as 
    \begin{subequations}\label{eq:interconnected_system_generic_err}
        \begin{align}
            \xtp &= \xt + \delta f(\xt,\twt + h(\xt))\label{eq:slow_system_generic_err}
            \\
            \twtp &= g(\twt  +  h(\xt),\xt,\delta) -  h(\xt)  +  \Delta h(\xtp,\xt),\label{eq:fast_system_generic_err}
        \end{align}
    \end{subequations}
    where $\Delta h(\xtp,\xt) \coloneqq - h(\xtp) + h(\xt)$. Pick $W$ as in~\eqref{eq:W_generic}. By evaluating $\Delta W(\xt) \coloneqq W(\xtp) - W(\xt)$ along the trajectories of~\eqref{eq:slow_system_generic_err}, we obtain 
    \begin{align}
        \Delta W(\xt) 
        &= W(\xt + \delta f(\xt,\twt + h(\xt))) - W(\xt)
        \notag\\
        &\stackrel{(a)}{=} 
        W(\xt + \delta f(\xt,h(\xt))) - W(\xt) 
        + W(\xt + \delta f(\xt,\twt + h(\xt))) 
        - W(\xt + \delta f(\xt,h(\xt)))
        \notag\\
        &\stackrel{(b)}{\leq}
         -\delta c_3\norm{\xt - \xstar}^2 + W(\xt + \delta f(\xt,\twt + h(\xt))) 
        - W(\xt + \delta f(\xt,h(\xt)))
        \notag\\
        &\stackrel{(c)}{\leq}
         -\delta c_3\norm{\xt - \xstar}^2 + 2\delta c_4 L_f\norm{\twt}\norm{\xt - \xstar} 
         + \delta^2c_4 L_f\norm{\twt}\norm{ f(\xt,\twt + h(\xt))}
        \notag\\
        &\hspace{.5cm}
        + \delta^2c_4L_f\norm{\twt} \norm{ f(\xt, h(\xt))},\label{eq:deltaU_generic}
    \end{align}
    where in $(a)$ we add and subtract the term $W(\xt + \delta f(\xt,h(\xt)))$, in $(b)$ we exploit~\eqref{eq:W_minus_generic} to bound the difference of the first two terms, in $(c)$ we use~\eqref{eq:W_bound_generic}, the Lipschitz continuity of $f$, and the triangle inequality. By recalling that $f(\xstar, h(\xstar)) = 0$ we can thus write 
    \begin{align}
        \| f(\xt, \twt + h(\xt))\| &= \norm{ f(\xt, \twt + h(\xt)) - f(\xstar, h(\xstar))}
        \notag\\
        &\stackrel{(a)}{\leq}
         L_f\norm{\xt - \xstar} + L_f\norm{\twt + h(\xt) - h(\xstar)}
         \nonumber \\
         &\stackrel{(b)}{\leq}
         L_f(1 + L_h)\norm{\xt - \xstar} + L_f\norm{\twt},\label{eq:bound_Lf_L_h_wt}
    \end{align}
    where in $(a)$ we use the Lipschitz continuity of $f$ and $h$, and in $(b)$ we use the Lipschitz continuity of $h$ together with the triangle inequality. With similar arguments, we have
    \begin{align}
        \norm{ f(\xt, h(\xt))} \leq L_f (1+ L_h)\norm{\xt - \xstar}.\label{eq:bound_L_f_L_h}
    \end{align}
Using inequalities~\eqref{eq:bound_Lf_L_h_wt} and~\eqref{eq:bound_L_f_L_h} we then bound~\eqref{eq:deltaU_generic} as 
    \begin{align}
        \Delta W (\xt) &\leq -\delta c_3\norm{\xt - \xstar}^2 + 2\delta c_4 L_f\norm{\twt}\norm{\xt - \xstar} 
        + \delta^2c_4 L_f^2\norm{\twt}^2
        \notag\\
        &\hspace{.5cm}
        +  2 \delta^2c_4 L_f^2(1+L_h)\norm{\twt}\norm{\xt - \xstar}
        \notag\\
        &\leq
        -\delta c_3\norm{\xt - \xstar}^2 + \delta^2 k_3\norm{\twt}^2
        + (\delta k_1 + \delta^2 k_2)\norm{\twt}\norm{\xt - \xstar}, \label{eq:deltaU_generic_final} 
    \end{align}
    where we introduce the constants
    \begin{align*}
        k_1 &\coloneqq 2c_4L_f,
        \quad
        k_2 \coloneqq 2c_4 L_f^2(1+L_h),
        \quad%
        k_3 \coloneqq c_4 L_f^2.
    \end{align*}

    We now pick $U$ as in~\eqref{eq:U_generic}. By evaluating $\Delta U(\twt) \coloneqq U(\twtp) - U(\twt)$ along the trajectories of~\eqref{eq:fast_system_generic_err}, we obtain 
    \begin{align}
        \Delta U(\tw) 
        &
        = U(g(\twt + h(\xt),\xt,\delta) -h(\xt) + \Delta h(\xtp,\xt)) 
        - U(\twt)
        \notag\\
        &\stackrel{(a)}{\leq}  U(g(\twt + h(\xt),\xt,\delta) - h(\xt)) - U(\twt) 
        - U(g(\twt + h(\xt),\xt,\delta) - h(\xt))
        \notag\\
        &\hspace{.5cm}
        + U(g(\twt + h(\xt),\xt,\delta) -h(\xt) + \Delta h(\xtp,\xt))
        \notag\\        
        &\stackrel{(b)}{\leq}  -b_3\norm{\twt}^2  - U(g(\twt + h(\xt),\xt,\delta) - h(\xt))
        \notag\\
        &\hspace{.5cm}
        + U(g(\twt + h(\xt),\xt,\delta) -h(\xt) + \Delta h(\xtp,\xt)) 
        \notag\\
        &\stackrel{(c)}{\leq}  -b_3\norm{\twt}^2  
        + b_4\norm{\Delta h(\xtp,\xt)}\norm{g(\twt + h(\xt),\xt,\delta) - h(\xt)+ \Delta h(\xtp,\xt)} 
        \notag\\
        &\hspace{.5cm}
        + b_4\norm{\Delta h(\xtp,\xt)}\norm{g(\twt + h(\xt),\xt,\delta) - h(\xt)}
        \notag\\
        &\stackrel{(d)}{\leq}  -b_3\norm{\twt}^2  + b_4\norm{\Delta h(\xtp,\xt)}^2 
        + 2 b_4\norm{\Delta h(\xtp,\xt)}\norm{g(\twt  +  h(\xt),\xt,\delta)  -  h(\xt)},\label{eq:DeltaW_generic}
    \end{align}
    where in $(a)$ we add and subtract $U(g(\twt + h(\xt),\xt,\delta)-h(\xt))$, in $(b)$ we exploit~\eqref{eq:U_minus_generic} to bound the first two terms, in $(c)$ we use~\eqref{eq:U_bound_generic} to bound the the difference of the last two terms, and $(d)$ uses the triangle inequality. By using the definition of $\Delta h(\xtp,\xt)$ and the Lipschitz continuity of $h$, we write
    \begin{align}
        \norm{\Delta h(\xtp,\xt)} &\leq L_h\norm{\xtp - \xt} 
        \notag\\
        &\stackrel{(a)}{\leq}
        \delta L_h\norm{f(\xt,\twt+h(\xt))}
        \notag\\
        &\stackrel{(b)}{\leq}
        \delta L_h\norm{f(\xt,\twt+h(\xt)) - f(\xstar, h(\xstar))}
        \notag\\
        &\stackrel{(c)}{\leq}
        \delta L_hL_f(1  +  L_h)\norm{\xt  -  \xstar}  +  \delta L_hL_f\norm{\twt},\label{eq:bound_Deltah}
    \end{align}
    where in $(a)$ we use the update~\eqref{eq:slow_system_generic_err}, in $(b)$ we add the term $f(\xstar,h(\xstar))$ since this is zero, and in $(c)$ we use the triangle inequality and the Lipschitz continuity of $f$ and $h$. Moreover, since $g(h(\xt),\xt,\delta) = h(\xt)$, we obtain
    \begin{align}
        &\norm{g(\twt + h(\xt),\xt,\delta) - h(\xt)} 
        =\norm{g(\twt + h(\xt),\xt,\delta)  -  g(h(\xt),\xt,\delta)}
         \leq  
        L_g\norm{\twt},\label{eq:bound_g}
    \end{align}
    where the inequality is due to the Lipschitz continuity of $g$. Using inequalities~\eqref{eq:bound_Deltah} and~\eqref{eq:bound_g}, we then bound~\eqref{eq:DeltaW_generic} as 
    \begin{align}
        \Delta U(\tw) 
        &\leq -b_3\norm{\twt}^2  + 2\delta b_4 L_h L_gL_f(1 + L_h)\norm{\xt - \xstar}\norm{\twt} 
        + 2\delta b_4 L_h L_gL_f\norm{\twt}^2
        \notag\\
        &\hspace{.5cm}
        +\delta^2b_4 L_h^2 L_f^2(1 + L_h)^2\norm{\xt - \xstar}^2 
        + 2\delta^2b_4 L_h^2 L_f^2(1 + L_h)\norm{\xt - \xstar}\norm{\twt} 
        \notag\\
        &\hspace{.5cm}
        + \delta^2 b_4 L_h^2 L_f^2\norm{\twt}^2
        \notag\\
        &\leq
        (-b_3 +\delta k_6 + \delta^2 k_7)\norm{\twt}^2 + \delta^2 k_8\norm{\xt - \xstar}^2
        + (\delta k_4 + \delta^2 k_5)\norm{\xt - \xstar}\norm{\twt}, \label{eq:deltaW_generic_final}
    \end{align}
    where we introduce the constants 
    \begin{align*}
        k_4 &\coloneqq 2b_4 L_hL_gL_f(1 + L_h), \quad
        &&k_5 \coloneqq 2b_4 L_h^2L_f^2(1 + L_h),
        \\
        k_6 &\coloneqq 2b_4 L_hL_gL_f,
        \quad
        &&k_7 \coloneqq b_4 L_h^2L_f^2,
        \\
        k_8 &\coloneqq b_4 L_h^2L_f^2(1+L_h)^2.
    \end{align*}
    We pick the following Lyapunov candidate $V: \mathcal{D} \times \R^m \to \R$:
    \begin{align*}
        V(\xt,\twt) = W(\xt) + U(\twt). 
    \end{align*}
    By evaluating $\Delta V(\xt,\twt) \coloneqq V(\xtp,\twtp) - V(\xt,\twt) = \Delta W(\xt) + \Delta U(\twt)$ along the trajectories of~\eqref{eq:interconnected_system_generic_err}, we can use the results~\eqref{eq:deltaU_generic_final} and~\eqref{eq:deltaW_generic_final} to write 
    \begin{align}
        \Delta V(\xt,\twt) &\leq -\begin{bmatrix}
            \norm{\xt - \xstar}\\
            \norm{\twt}
        \end{bmatrix}\T Q(\delta) \begin{bmatrix}
            \norm{\xt - \xstar}\\
            \norm{\twt}
        \end{bmatrix},\label{eq:deltaV_generic}
    \end{align}
    where we define the matrix $Q(\delta) = Q(\delta)\T \in \R^2$ as 
    \begin{align*}
        &Q(\delta) \coloneqq \begin{bmatrix}
            \delta c_3 -\delta^2k_8& q_{21}(\delta)\\
            q_{21}(\delta) & b_3 - \delta k_6 - \delta^2 (k_3 + k_7)
        \end{bmatrix},
    \end{align*}
    with $q_{21}(\delta) \coloneqq -\frac{1}{2}(\delta (k_1+k_4) + \delta^2 (k_2+k_5))$. By relying on the Sylvester criterion \cite{khalil2002nonlinear}, we know that $Q \succ 0$ if and only if 
    \begin{align}\label{eq:conditions_generic}
        c_3b_3 > p(\delta)
    \end{align}
    where the polynomial $p(\delta)$ is defined as 
    \begin{align}
        p(\delta) &\coloneqq q_{21}(\delta)^2 + \delta^2 c_3k_6 
        + \delta^2 (\delta c_3(k_3 + k_7)+b_3k_8)
          -\delta^3k_6k_8 - \delta^4k_8(k_3+k_7).
    \end{align}
    We note that $p$ is a continuous function of $\delta$ and $\lim_{\delta \to 0}p(\delta) = 0$. Hence, there exists some $\bar{\delta} \in (0,\min\{\bar{\delta}_1,\bar{\delta}_2\})$ -- recall that $\bar{\delta}_1$ and $\bar{\delta}_2$  exist as $U$ and $W$ are taken to satisfy \eqref{eq:U_generic} and \eqref{eq:W_generic} -- so that~\eqref{eq:conditions_generic} is satisfied for any $\delta \in (0,\bar{\delta})$. Under such a choice of $\delta$, and denoting by $q>0$ the smallest eigenvalue of $Q(\delta)$, we can bound~\eqref{eq:deltaV_generic} as 
    \begin{align*}
        \Delta V(\xt,\twt) \leq -q\norm{\begin{bmatrix}\norm{\xt - \xstar}\\
            \norm{\twt}\end{bmatrix}}^2,
    \end{align*}
    which allows us to conclude, in view of \cite[Theorem~13.2]{chellaboina2008nonlinear}, that $(\xstar, 0)$ is an exponentially stable equilibrium point for system~\eqref{eq:interconnected_system_generic_err}. The theorem's conclusion follows then by considering the definition of exponentially stable equilibrium point and by reverting to the original coordinates $(\xt,\wt)$.

\subsection{Proofs of technical lemmas of Section~\ref{sec:SP}}
\label{sec:proof_SP}

\emph{Proof of Lemma~\ref{lemma:bl}}: system~\eqref{eq:bl} is a linear autonomous system whose state matrix $\Rd\T \Wd\Rd \in \R^{(N-1)d \times (N-1)d}$ is Schur. 
Hence, there exists $P \in \R^{(N-1)d \times (N-1)d}$, $P = P\T \succ 0$ for the candidate Lyapunov function $U(\tzt) = (\tzt)\T P\tzt$, solving the Lyapunov equation
\begin{equation}
    (\Rd\T \Wd\Rd)\T P \Rd\T \Wd\Rd - P = -Q. \label{eq:discr_lyap}
\end{equation}
for any $Q \in \R^{(N-1)d \times (N-1)d}$, $Q = Q\T \succ 0$.
Condition~\eqref{eq:U_first_bound_generic} follows then from the fact that $U$ is quadratic with $P \succ 0$ so $b_1$ and $b_2$ can be chosen to be its minimum and maximum eigenvalue, respectively.
The left-hand side of \eqref{eq:U_minus_generic} becomes $(\tzt)\T  ((\Rd\T \Wd\Rd)\T P \Rd\T \Wd\Rd - P) \tzt = -(\tzt)\T  Q \tzt$, where the equality is due to \eqref{eq:discr_lyap}. Hence,~\eqref{eq:U_minus_generic} is satisfied by taking $b_3$ to be the smallest eigenvalue of $Q$. To see \eqref{eq:U_bound_generic} notice that
\begin{align}
\norm{U(\tzt_1)-U(\tzt_2)} &= \norm{(\tzt_1)\T P\tzt_1 - (\tzt_2)\T P\tzt_2} \notag\\
&\stackrel{(a)}{\leq} \norm{(\tzt_1)\T P\tzt_1 - (\tzt_1)\T P\tzt_2} + \norm{(\tzt_2)\T P\tzt_1 - (\tzt_2)\T P\tzt_2} \notag\\
&\stackrel{(b)}{\leq} \norm{P} \norm{\tzt_1 - \tzt_2} \norm{\tzt_1} + \norm{P} \norm{\tzt_1 - \tzt_2} \norm{\tzt_2},\label{eq:eq_growth}
\end{align}
where $(a)$ follows from adding and subtracting $(\tzt_1)\T P\tzt_2$ and using the triangle inequality, while $(b)$ from the Cauchy-Schwarz inequality.
The bound~\eqref{eq:U_bound_generic} follows from~\eqref{eq:eq_growth} by setting $b_4$ as the largest eigenvalue of $P$.\hfill\qedsymbol

\smallskip

    We provide here the following technical lemma which is used in the proof of Lemma~\ref{lemma:rs}.

\begin{lemma}[Contraction of strongly monotone operator]\label{lemma:contraction}
    Let $F :\R^n \to \R^n$ be $\mu$-strongly monotone and $L$-Lipschitz continuous. If $\gamma \in (0,2\mu/L^2)$, then for any $\x, \x^\prime \in \R^n$ it holds %
    \begin{align*}
        \norm{\x - \gamma F(\x) - \x^\prime + \gamma F(\x^\prime)} \leq (1 - \bar{\mu})\norm{\x - \x^\prime},
    \end{align*}
    where $\bar{\mu} \coloneqq 1 - \sqrt{1 - \gamma(2\mu - \gamma L^2)} \in (0,1]$.
\end{lemma}
\begin{proof}
    We have that
    \begin{align}
        \norm{\x - \gamma F(\x) - \x^\prime + \gamma F(\x^\prime)}^2 
        &= \norm{\x - \x^\prime}^2 + \gamma^2\norm{F(\x) - F(\x^\prime)}^2
        - 2\gamma(\x - \x^\prime)^\top(F(\x) - F(\x^\prime))
        \notag\\
        &\stackrel{(a)}{\leq}
        \norm{\x - \x^\prime}^2 - \gamma(2\mu - \gamma L^2)\norm{\x - \x^\prime}^2,\label{eq:intermediate_contraction}
    \end{align}
    where in $(a)$ we use the strong monotonicity and the Lipschitz continuity of $F$. By construction, $\bar{\mu} \in (0,1]$ is equivalent to $\gamma(2\mu - \gamma L^2) > 0$ and $\gamma(2\mu - \gamma L^2) \leq 1$. The former holds since $\gamma \in (0,2\mu/L^2)$. To see the latter, notice that, by definition of $\mu$-strong monotonicity and $L$-Lipschitz continuity, we have
    		\begin{align*}
    			\mu \norm{\x-\x^\prime}^2 
                &\leq (F(\x) - F(\x^\prime))\T (\x-\x^\prime) 
                \\
	    		&\leq \norm{F(\x) - F(\x^\prime)}\norm{\x-\x^\prime} \leq L \norm{\x-\x^\prime}^2,
    		\end{align*}
    for any $\x,\x^\prime$, hence $\mu \leq L$. Thus, for any $\gamma$, it holds that $1 - 2\mu\gamma + \gamma^2 L^2 \geq 1 - 2\gamma L+ \gamma^2 L^2 = (1-\gamma L)^2 \geq 0$.
\end{proof}

\smallskip

\emph{Proof of Lemma~\ref{lemma:rs}}: pick $W: \R^{n} \to \R$ defined as 
        \begin{align*}
            W(\xt) = \frac{1}{2}\norm{\xt - \xstar}^2.
        \end{align*}
        Since $W$ is a quadratic function, conditions~\eqref{eq:W_first_bound_generic} and~\eqref{eq:W_bound_generic} are satisfied. To show \eqref{eq:W_minus_generic} we evaluate $\Delta W(\xt) \coloneqq W(\xtp) - W(\xt)$ along~\eqref{eq:rs_explicit}. We then have
        \begin{align}
            \Delta W(\xt) 
            &= \frac{1}{2}\norm{(1 - \delta)\xt  +  \delta\left(\Px{\xt - \gamma F(\xt)}\right)  -  \xstar}^2 
             -  \frac{1}{2}\norm{\xt  -  \xstar}^2
            \notag\\
            &\stackrel{(a)}{\leq}
            \frac{(1-\delta)^2}{2}\norm{\xt-\xstar}^2  - \frac{1}{2}\norm{\xt - \xstar}^2
            \notag\\
            &\hspace{.5cm}
            + (\delta  -  \delta^2)\norm{\xt  - \xstar }\norm{\Px{\xt  -  \gamma F(\xt)}  -  \Px{\xstar  -  \gamma F(\xstar)}} 
            \notag\\
            &\hspace{0.4cm}
            + \frac{\delta^2}{2}\norm{\Px{\xt - \gamma F(\xt)} - \Px{\xstar - \gamma F(\xstar)}}^2 
            \notag\\
            &\stackrel{(b)}{\leq}
            \frac{(1-\delta)^2}{2}\norm{\xt-\xstar}^2 - \frac{1}{2}\norm{\xt - \xstar}^2
            \notag\\
            &\hspace{0.4cm}
            + (\delta-\delta^2)\norm{\xt-\xstar}\norm{\xt - \gamma F(\xt) - \xstar + \gamma F(\xstar)} 
            \notag\\
            &\hspace{0.4cm}
            + \frac{\delta^2}{2}\norm{\xt - \gamma F(\xt) - \xstar + \gamma F(\xstar)}^2,\label{eq:deltaV_uncon}
        \end{align} 
        where in $(a)$ we introduce $\delta(\xstar - \Px{\xstar - \gamma F(\xstar)})$ within the first norm, as this is zero due to the definition of $\xstar$, expand the square, and use the Cauchy-Schwarz inequality. Inequality $(b)$ follows by the fact that for any $a,b$, we have that $\norm{\Px{a}-\Px{b}} \leq \norm{a-b}$, since the projection operator is nonexpansive.
        Since $F$ is $\mu$-strongly monotone and $\beta_1$ Lipschitz continuous (cf. Standing Assumption~\ref{ass:objective_function}), set $\bar{\gamma} = 2\mu/(\beta_1)^2$ and choose $\gamma \in (0,\bar{\gamma})$. Applying Lemma~\ref{lemma:contraction} yields
        \begin{align}
            \norm{\xt - \gamma F(\xt) - \xstar + \gamma F(\xstar)} \leq (1 - \bar{\mu})\norm{\xt - \xstar},\label{eq:strongly_monotone_constrained_uncon}
        \end{align}
        with $\bar{\mu} = 1 - \sqrt{1 - \gamma(2\mu - \gamma (\beta_1)^2)} \in (0,1]$. Thus, by using the inequality in~\eqref{eq:strongly_monotone_constrained_uncon}, we can bound~\eqref{eq:deltaV_uncon} as follows
        \begin{align}
            \Delta W(\xt) &\leq \frac{(1-\delta)^2}{2}\norm{\xt-\xstar}^2 - \frac{1}{2}\norm{\xt - \xstar}^2
            + (\delta-\delta^2)(1 - \bar{\mu})\norm{\xt-\xstar}^2
            \notag\\
            &\hspace{.5cm}
            + \delta^2(1 - \bar{\mu})^2/2\norm{\xt-\xstar}^2
            \notag\\
            &\stackrel{(a)}{=}
            -\delta\bar{\mu}\left(1-\delta\bar{\mu}/2\right)\norm{\xt - \xstar}^2.\label{eq:deltaV_final_uncon}
        \end{align}
        where $(a)$ is obtained by rearranging the above terms.
        Thus, for any $\delta \in (0,\bar{\delta}_2)$ with $\bar{\delta}_2 \coloneqq 2/\bar{\mu}$, $(1-\delta\bar{\mu}/2)>0$ in~\eqref{eq:deltaV_final_uncon}, thus establishing condition~\eqref{eq:W_minus_generic} and concluding the proof.

\subsection{Proofs of technical lemmas of Section~\ref{sec:sp_primal_dual}}
\label{sec:proof_SP_primal_dual}

\emph{Proof of Lemma~\ref{lemma:bl_primal_dual}}: since $\Rm\T \onem = 0$, we can write
        \begin{align}
            &
            \Rm\T\Gl\left(\onem(A\x -b) + \Rm\tpyt, \onem\bl + \Rm\tplt\right) 
            \notag\\
            &= \Rm\T\bigg(\Gl\left(\onem(A\x -b) + \Rm\tpyt, \onem\bl + \Rm\tplt\right) 
            - \onem \nabla_{\lambda}H(A\x - b,\bl)\bigg)
            \notag\\
            &=
            \Rm\T\bigg(\Gl\left(\onem(A\x -b) + \Rm\tpyt, \onem\bl + \Rm\tplt\right) 
            - \Gl(\onem(A\x - b), \onem\bl)\bigg),\label{eq:tildeu}
        \end{align}
        where in the last equality we used $\onem \nabla_{\lambda}H(A\x - b,\bl) = \Gl(\onem(A\x - b), \onem\bl)$. 
        Following \cite[Lemma~3]{qu2018exponential}, notice that, for any $r_1, r_2 \in \R$, there exists $\epsilon(r_1,r_2) \in [0,1]$ so that\footnote{If $r_1 \ne r_2$, pick $\epsilon = \frac{\max\{r_1,0\} - \max\{r_2,0\}}{r_1 - r_2}$, otherwise set $\epsilon = 0$.}
        \begin{equation}
            \max\{r_1,0\} - \max\{r_2,0\} = \epsilon(r_1,r_2)(r_1 - r_2).\label{eq:max_r1_r2_bl} 
        \end{equation}
        Let us introduce 
        \begin{align}\label{eq:qit_sit}
            q_i^\iter &:= \sum_{\ell=1}^m [\Rm\tpyt]_{\ell + (i-1)m}e_\ell,
            \quad 
            p_i^\iter := \sum_{\ell=1}^m [\Rm\tplt]_{\ell + (i-1)m}e_\ell,
        \end{align}
        and use them to define
        \begin{align}\label{eq:r1ir2i}
            \begin{split}
            r_{1,i}^\iter &:=  \rho(A\x -b + q_i^\iter) +  \bl + p_i^\iter
            \\
            r_{2,i} &:=  \rho(A\x - b) + \bl.
        \end{split}
        \end{align}
        By the definition of $\tilde{u}(\chi,\twt)$ we have that its norm $\norm{\tilde{u}(\chi,\twt)}$ is equal to the norm of the quantity in \eqref{eq:tildeu}.
        Let $\nw:= (N-1)(2m + d)$.
        As such, for any $\chi \in \R^{n+m}$ and $\twt \in \R^{\nw}$, we use the definition of $\Gl$ in~\eqref{eq:eqG}, $r_{1,i}^\iter$ and $r_{2,i}$ in~\eqref{eq:r1ir2i}, and apply~\eqref{eq:max_r1_r2_bl} for each component of $\tilde{u}(\chi,\twt)$ obtaining
        \begin{align}
            \norm{\tilde{u}(\chi,\twt)}
            &\leq 
             \bigg\|\Rm\T \frac{1}{\rho}\col\bigg(\sum_{\ell=1}^m \epsilon([r_{1,i}^\iter]_\ell,[r_{2,i}]_\ell)\big([r_{1,i}^\iter - \bl - p_i^\iter]_{\ell} 
             - [r_{2,i} - \bl]_{\ell}\big) e_\ell\bigg)_{i=1}^N \bigg\|
             \notag\\
             &\stackrel{(a)}{\leq}
             \bigg\|\Rm\T \frac{1}{\rho}\col\left(\sum_{\ell=1}^m \big([r_{1,i}^\iter - \bl - p_i^\iter]_{\ell} 
             - [r_{2,i} - \bl]_{\ell}\big) e_\ell\right)_{i=1}^N \bigg\|
             \notag\\
             &\stackrel{(b)}{=} \norm{\Rm\T \frac{1}{\rho}\col\left(\sum_{\ell=1}^m \rho[q_i^\iter]_\ell e_\ell\right)_{i=1}^N}
             \notag\\
             &\stackrel{(c)}{=}  \norm{\Rm\T \Rm\tpyt}
             \stackrel{(d)}{\leq}\norm{\twt}
             ,\label{eq:bound_tilde_u}
        \end{align}
            where in $(a)$ we use the fact that $\epsilon([r_{1,i}^\iter]_\ell,[r_{2,i}]_\ell) \in [0,1]$ for all $\ell \in\until{m}$ and $i \in \mc I$, $(b)$ uses the definitions in~\eqref{eq:r1ir2i} to simplify the terms, $(c)$ follows from~\eqref{eq:qit_sit}, and $(d)$ uses $\Rm\T \Rm = I$ and $\norm{\tpyt} \leq \norm{\twt}$ that holds since $\tpyt$ is a component of $\twt$.     
Now, let $U: \R^{\nw} \to \R$ be
        \begin{align*}
            U(\tw) = (\tw)\T \cU \tw,
        \end{align*}
        where $\cU \in \R^{\nw \times \nw}$ with $\cU = \cU\T \succ 0$, such that
        \begin{align}\label{eq:lypaunov_bl_primal_dual}
            S\T \cU S - \cU = -I.
        \end{align} 
        We remark that such a matrix $\cU$ always exists because, in light of Standing Assumption~\ref{ass:network}, both $\Rd\T \Wd\Rd$ and $\Rm\T \Wm\Rm$ are Schur matrices and, thus, $S$ is Schur as well. Under this choice of $U$, conditions \eqref{eq:U_first_bound_generic} and~\eqref{eq:U_bound_generic} are satisfied. To show \eqref{eq:U_minus_generic} we evaluate $\Delta U(\twt) \coloneqq U(\tw^{t+1}) - U(\twt)$ along the trajectories of~\eqref{eq:bl_primal_dual}, obtaining
        \begin{align}
            \Delta U(\twt) 
            &= (S\twt + \delta \tilde{u}(\chi,\twt))\T \cU (S\twt + \delta \tilde{u}(\chi,\twt))
            - (\twt)\T \cU \twt \notag
       \\
            &=-\norm{\twt}^2 + 2\delta (\twt)\T S\T \cU \tilde{u}(\chi,\twt) 
            + \delta^2  \tilde{u}(\chi,\twt)\T  \cU \tilde{u}(\chi,\twt)
            \notag
            \\
            &\leq 
            - (1 - \delta \mu_1 - \delta^2 \mu_2)\norm{\twt}^2,\label{eq:DeltaU_bl_primal_dual}
        \end{align}
        where the second equality is due to \eqref{eq:lypaunov_bl_primal_dual}, and the inequality follows from \eqref{eq:bound_tilde_u} and the Cauchy-Schwarz inequality , with the constants $\mu_1 \coloneqq 2\norm{S}\norm{\cU}$ and $\mu_2 \coloneqq \norm{\cU}$. 
        Thus, there always exists $\bar{\delta}_1 > 0$ small enough so that $(1 - \delta \mu_1 - \delta^2 \mu_2) > 0$ for any $\delta \in (0,\bar{\delta}_1)$, concluding the proof.
		
		\smallskip

\emph{Proof of Lemma~\ref{lemma:rs_primal_dual}}: the proof is inspired by~\cite[Theorem~2, Lemma~3, Lemma~4]{qu2018exponential},  adapted to our framework. Let $\cF: \R^{n+m}\to\R^{n+m}$ and $\mathcal{H}: \R^{n+m}\to\R^{n+m}$ be defined as
\begin{subequations}\label{eq:cFcH}
\begin{align}
        \cF(\chit) &\coloneqq 
        \begin{bmatrix}
            F\left(\begin{bmatrix}
            I& 0
        \end{bmatrix}\chit\right)
        \\
        0
    \end{bmatrix},
    \\
        \mathcal{H} (\chit) &\coloneqq \begin{bmatrix}\nabla_\x H\left(A\begin{bmatrix}
            I& 0
        \end{bmatrix}\chit - b,\begin{bmatrix}
            0& I
        \end{bmatrix}\chit\right)
        \\
        -\nabla_\lambda H\left(A\begin{bmatrix}
            I& 0
        \end{bmatrix}\chit - b,\begin{bmatrix}
            0& I
        \end{bmatrix}\chit\right)
        \end{bmatrix}. \label{eq:cFcH_b}
\end{align}
\end{subequations}
Applying \eqref{eq:max_r1_r2_bl} to each of the components of $\mc H(\chit)- \mc H(\chistar)$, for any $\chit \in \R^{n + m}$ we obtain
\begin{align}
    &\mc H(\chit)- \mc H(\chistar)
     =\begin{bmatrix}
        \rho A\T E(\chit,\chistar)A&  A\T E(\chit,\chistar)
        \\
        -E(\chit,\chistar)A& -\frac{1}{\rho}(E(\chit,\chistar) - I)
    \end{bmatrix}(\chit  - \chistar),\label{eq:nablaH}
\end{align}
where $E(\chit,\chistar) \coloneqq  \diag(\epsilon_1(\chit,\chistar),\dots,\epsilon_m(\chit,\chistar))$ and $\epsilon_\ell(\chit,\chistar) \in [0,1]$ so that 
\begin{align*}
    & \max\{\rho[A\xt  -  b]_\ell  +  [\blt]_\ell,0\}  -  \max\{\rho[A\xstar  -  b]_\ell  +  [\lstar]_\ell,0 \}
    \notag\\
    &= \epsilon_\ell(\chit,\chistar)(\rho[A\xt - b - A\xstar - b]_\ell + [\blt]_\ell - [\lstar]_\ell),
\end{align*}
for all $\ell \in \until{m}$ and $\chit \coloneqq \col(\xt,\blt) \in \R^{n+m}$. 
Moreover, for any $\xt \in \R^n$, we have
\begin{align}
    F(\xt) - F(\xstar) &= \int_{0}^1 \nabla F((1-\nu)\xstar + \nu\xt)(\xt - \xstar)d\nu
    \notag\\
    &\stackrel{(a)}{=} \left[\int_{0}^1 \nabla F((1 - \nu)\xstar + \nu\xt)d\nu\right](\xt - \xstar)
    \notag\\
    &\stackrel{(b)}{=}  B(\xt,\xstar)(\xt - \xstar).\label{eq:integral}
\end{align}
where in $(a)$ we have extracted the term $(\xt - \xstar)$ from the integral and in $(b)$ we have introduced $B(\xt,\xstar) \coloneqq \int_{0}^1 \nabla F((1-\nu)\xstar + \nu\xt)d\nu$. Since $F$ is $\mu$-strongly monotone and $\beta_1$-Lipschitz continuous (cf. Standing Assumption~\ref{ass:objective_function}),  we can uniformly bound the integrand term of~\eqref{eq:integral} as
\begin{equation*}
    \mu I \preccurlyeq \nabla F((1-\nu)\xstar + \nu\xt) \preccurlyeq \beta_1 I,
\end{equation*} 
which leads to 
\begin{equation}
    \mu I \preccurlyeq \int_{0}^1 \mu I d\nu \preccurlyeq B(\xt,\xstar) \preccurlyeq \int_{0}^1 \beta_1 I d\nu \preccurlyeq \beta_1 I.\label{eq:B_bounds}
\end{equation}
Combining~\eqref{eq:f_primal_dual},~\eqref{eq:cFcH},~\eqref{eq:nablaH}, and~\eqref{eq:integral}, we can write 
\begin{align}
    f(\chit,h(\chit)) - f(\chistar,h(\chistar))
    &= -\cF(\chit) + \cF(\chistar) - (\mc H(\chit) - \mc H(\chistar)) 
    \notag\\
    &
    = D(\chit,\chistar)(\chit - \chistar),\label{eq:D}
\end{align}
where $D(\chit,\chistar) \in \R^{(n + m) \times (n + m)}$ is given by 
\begin{align*}
    &D(\chit,\chistar) 
    \coloneqq \begin{bmatrix}
        -B(\chit,\chistar) -\rho A\T E(\chit,\chistar)A& -A\T E(\chit,\chistar)\\
        E(\chit,\chistar)A& \frac{1}{\rho}(E(\chit,\chistar) - I)
    \end{bmatrix}.
\end{align*}
We then have that for any $\chit \in \R^{n + m}$,
\begin{align}\label{eq:G_upper_bound}
    \norm{D(\chit,\chistar)(\chit - \chistar)}^2_M \leq \mu_1\norm{\chit - \chistar}^2_M,
\end{align}
where $\mu_1 \coloneqq \left(\max\left\{\beta_1 + \rho\norm{A}^2,\frac{1}{\rho}\right\}\right)^2$ and the inequality follows by inspection of $D(\chit,\chistar)(\chit - \chistar)$ and using $\norm{E(\chit,\chistar)} \leq 1$.
Now, let $W: \R^{n+m} \to \R$ be defined as 
\begin{equation}\label{eq:W_rs_primal_dual}
W(\chi) = (\chi - \chistar)\T M(\chi - \chistar),
\end{equation}
where $M \in \R^{(n+m) \times (n+m)}$ is defined as 
\begin{align}\label{eq:P_li_na}
    M \coloneqq \begin{bmatrix}
         cI& A\T\\
        A& cI
    \end{bmatrix}.
\end{align}
Note that $M \succ 0$ for any $c > \sqrt{\kappa_2}$ and, thus, $W$ satisfies~\eqref{eq:W_first_bound_generic} and~\eqref{eq:W_bound_generic}. 
To show \eqref{eq:W_minus_generic}, we evaluate $\Delta W(\chit) \coloneqq W(\chitp) - W(\chit)$ along the trajectories of~\eqref{eq:rs_primal_dual}, obtaining
\begin{align}
&\Delta W(\chit) 
\notag\\
&= \norm{\chit + \delta f(\chit,h(\chit)) - \chistar}_M^2 - \norm{\chit - \chistar}_M^2
\notag\\
&\stackrel{(a)}{=}
\norm{\chit + \delta f(\chit,h(\chit)) - \chistar - \delta f(\chistar, h(\chistar))}_M^2 
- \norm{\chit - \chistar}_M^2
\notag\\
&\stackrel{(b)}{=}
\norm{\chit -\chistar + \delta D(\chit,\chistar)(\chit - \chistar)}_M^2 - \norm{\chit - \chistar}_M^2
\notag\\
&\stackrel{(c)}{=}
\delta (\chit - \chistar)\T (D(\chit,\chistar)\T M + M D(\chit,\chistar))(\chit - \chistar) 
+ \delta^2 \norm{D(\chit,\chistar)(\chit - \chistar)}_M^2
,\label{eq:DeltaW_final}
\end{align}
where $(a)$ uses the fact that $f(\chistar,h(\chistar)) = 0$ (cf.~\eqref{eq:equilibrium_primal_dual}), $(b)$ rewrites the quantities using~\eqref{eq:D}, and $(c)$ expands $\norm{\cdot}_M^2$.
As in~\cite[Lemma~4]{qu2018exponential}, since it holds~\eqref{eq:B_bounds}, there exists $\bar{c} > 0$ such that, for any $c > \bar{c}$,
it holds
\begin{equation}\label{eq:result_lemma_li_na}
    D(\chit,\chistar)\T M + MD(\chit,\chistar) \leq -\tau M,
\end{equation}
where $\tau \coloneqq \frac{\kappa_1}{2c}$.
Therefore, by using~\eqref{eq:result_lemma_li_na} and~\eqref{eq:G_upper_bound}, we bound the right-hand side of~\eqref{eq:DeltaW_final} as 
\begin{equation*}
    \Delta W(\chit) \leq -\delta\left(\tau - \delta\mu_1\right)\norm{\chit - \chistar}_M^2. 
\end{equation*}

Setting $\bar{\delta} \coloneqq \frac{\tau}{\mu_1}$, \eqref{eq:DeltaW_final} ensures that for any $\delta \in (0,\bar{\delta})$, $W$ satisfies~\eqref{eq:W_minus_generic}, and the proof follows.

\end{document}